\documentclass[pra,onecolumn,superscriptaddress,nofootinbib]{revtex4}
\usepackage[a4paper, left=0.8in, right=0.8in, top=1.0in, bottom=1.0in]{geometry}

\usepackage{amsthm}
\usepackage{amsmath}
\usepackage{mathrsfs}
\usepackage{amssymb}
\usepackage{wasysym}
\usepackage{graphicx}
\usepackage{varwidth}
\usepackage{url}
\usepackage{dsfont}
\usepackage{float}
\usepackage{bm}
\usepackage{ifthen}
\usepackage[usenames,dvipsnames]{color}
\usepackage{mathrsfs}
\usepackage{hyperref}
\usepackage{float}
\usepackage{afterpage}
\usepackage{subcaption}

\newsavebox{\twosubbox}

\usepackage{physics}
\usepackage{xcolor}

\usepackage{algorithm,algpseudocode}
\algnewcommand{\LeftComment}[1]{\Statex \hspace{3em}\(\triangleright\) #1}

\usepackage{amsthm}
\newtheorem{theorem}{Theorem}[section]
\newtheorem{lemma}[theorem]{Lemma}
\newtheorem{corollary}[theorem]{Corollary}
\newtheorem{remark}[theorem]{Remark}
\theoremstyle{definition}
\newtheorem{definition}{Definition}[section]
\newtheorem{assumption}[definition]{Assumption}

\begin{document}

\title{A quantum algorithm for training wide and deep classical neural networks}

\author{Alexander Zlokapa}
\email{azlokapa@mit.edu}
\affiliation{Department of Physics, Massachusetts Institute of Technology, Cambridge, Massachusetts 02139}
\affiliation{Google Quantum AI, Venice, California 90291}
\author{Hartmut Neven}
\affiliation{Google Quantum AI, Venice, California 90291}
\author{Seth Lloyd}
\affiliation{Department of Mechanical Engineering and Research Laboratory of Electronics, Massachusetts Institute of Technology, Cambridge, Massachusetts 02139}

\begin{abstract}
Given the success of deep learning in classical machine learning, quantum algorithms for traditional neural network architectures may provide one of the most promising settings for quantum machine learning. Considering a fully-connected feedforward neural network, we show that conditions amenable to classical trainability via gradient descent coincide with those necessary for efficiently solving quantum linear systems. We propose a quantum algorithm to approximately train a wide and deep neural network up to $O(1/n)$ error for a training set of size $n$ by performing sparse matrix inversion in $O(\log n)$ time. To achieve an end-to-end exponential speedup over gradient descent, the data distribution must permit efficient state preparation and readout. We numerically demonstrate that the MNIST image dataset satisfies such conditions; moreover, the quantum algorithm matches the accuracy of the fully-connected network. Beyond the proven architecture, we provide empirical evidence for $O(\log n)$ training of a convolutional neural network with pooling.
\end{abstract}

\maketitle

\setcounter{secnumdepth}{0}

Due to the widespread success of classical machine learning, quantum machine learning has rapidly become a central topic of interest for future applications of quantum computing. While near-term approaches to quantum machine learning typically exploit the high dimensionality of the Hilbert space through kernel methods~\cite{Havlicek2019,Blank2020,PhysRevLett.122.040504,lloyd2020quantum,schuld2021quantum}, most exponential speedups require fault-tolerant quantum computers with access to a quantum random access memory (QRAM). Essential primitives in linear algebra and optimization~\cite{PhysRevLett.103.150502,doi:10.1137/16M1087072,Lloyd2014,PhysRevA.97.012327,8104077,Rebentrost_2019} have provided the basis for quantum analogues to common classical approaches to classification, clustering, regression, and other tasks in data analysis~\cite{lloyd2013quantum,Biamonte2017,PhysRevLett.113.130503,PhysRevA.94.022342,NEURIPS2019_16026d60}. However, while neural networks represent state-of-the-art classical machine learning for a variety of benchmark tasks, existing proposals for quantum neural networks lack a clear demonstration of a quantum speedup for tasks on classical datasets~\cite{PhysRevX.8.021050,farhi2018classification,PhysRevResearch.1.033063,PhysRevLett.121.040502}. In particular, \emph{deep} neural networks are empirically observed to achieve successful results in classical machine learning~\cite{simonyan2015deep,Szegedy_2015_CVPR,He_2016_CVPR}. Amidst rising computational requirements of deep learning due to larger datasets and neural network architectures~\cite{strubell-etal-2019-energy}, the use of quantum computers to efficiently train deep neural networks remains a fundamental issue in quantum machine learning.

While algorithms for quantum machine learning are largely based on methods from linear algebra, neural networks rely on nonlinearity to act as a universal approximator~\cite{Cybenko1989,LESHNO1993861}. Recent work on the dynamics of \emph{wide} neural networks --- i.e., where each of $L$ hidden layers has $m$ neurons for large $m$ --- have introduced the \emph{neural tangent kernel} (NTK), which represents such overparameterized neural networks as linearized models applied to nonlinear features~\cite{NEURIPS2018_5a4be1fa,NEURIPS2019_0d1a9651}. For example, one may consider a vanilla feedforward fully-connected neural network $f$ on a test data example $\mathbf x_*$. The output $f(\mathbf x_*)$ of the neural network is parameterized by weight matrices $W_i$ on the $i$th hidden layer and output layer weights $v$, giving
\begin{align}
\label{mt:eq:nn}
f(\mathbf{x}_*; W_1, \dots, W_L, v) := v \cdot \frac{1}{\sqrt{m}} \sigma\left(W_L \frac{1}{\sqrt{m}} \sigma\left(W_{L-1} \dots \frac{1}{\sqrt{m}} \sigma\left(W_1 \mathbf{x}_*\right)\dots\right) \right)
\end{align}
for nonlinear activation function $\sigma$. The NTK defines a kernel $k(\mathbf x_i, \mathbf x_j)$ between any pair of data examples. If the neural network is initialized with Gaussian-distributed weights and trained via gradient descent with squared loss, the expected output of the fully trained neural network on the test data point $\mathbf x_*$ is requires solving a system of $n$ linear equations for a training set of $n$ examples. Written explicitly for a training set $S = \{(\mathbf x_i, y_i) \in \mathbb{R}^d \times \{-1, 1\}\}_{i=1}^n$, the trained neural network output is given by
\begin{align}
\label{mt:eq:f-ntk}
    \mathbb{E}[f(\mathbf x_*)] = \mathbf{k}_*^T K^{-1}\mathbf y,
\end{align}
where $\mathbf y$ is the vector of labels over the training set, $\mathbf{k}_* \in \mathbb{R}^n$ is the vector generated by evaluating the NTK between $\mathbf x_*$ and all $\mathbf x_i$ in the training set $S$, and the $n \times n$ matrix $K$ corresponds to the NTK evaluated between all pairs of training data $\mathbf x_i, \mathbf x_j \in S$. Hence, the output of the trained neural network is determined by the evaluation of a matrix inversion and inner product.

Although we only provide a rigorous theoretical treatment of the properties of a fully-connected neural network (Eq.~\ref{mt:eq:nn}), the linear form of the NTK shown in Eq.~\ref{mt:eq:f-ntk} offers a general framework for neural network architectures similar to those used in state-of-the-art applications of deep learning~\cite{NIPS2014_81ca0262,pmlr-v119-shankar20a,NEURIPS2019_663fd3c5,yang2020tensor}. In particular, we provide numerical results for architectures that include convolutional and pooling layers, which are observed to satisfy the same properties as the fully-connected neural network.

The NTK formalism sheds light on the benefit of \emph{deep} neural networks, i.e. increasing the number of hidden layers $L$. As the neural network is deepened, the NTK matrix becomes increasingly well-conditioned, improving the speed at which gradient descent trains the neural network~\cite{agarwal2020deep}. In the regime of efficient training by gradient descent, the evaluation of $\mathbf{k}_*^T K^{-1}\mathbf y$ has a condition number that approaches unity as the dataset increases in size. As discussed by~\citet{Aaronson2015}, quantum algorithms to perform linear algebra operations often face limitations due to stringent theoretical caveats including matrix sparsity and well-conditioning. Hence, the well-conditioning of the NTK required for convergence by gradient descent corresponds to a necessary condition for an efficient quantum algorithm to train the neural network.

Our main result is a quantum algorithm to train a wide and deep neural network under an approximation of the NTK, estimating the trained neural network output with vanishing error as the training set size increases. We provide two different approximations: a \emph{sparsified} NTK and a \emph{diagonal} NTK. In both cases, convergence of the approximation to the exact neural network is guaranteed by matrix element bounds of the NTK; the same bounds also directly enable efficient gradient descent, highlighting the correspondence between conditions for trainable classical neural networks and an efficient quantum algorithm. The sparsified NTK approximation has a strictly tighter upper bound on the error compared to the diagonal NTK approximation, and numerical experiments confirm the better performance of the sparsified NTK.

Given a training set with $n$ data examples, both approximations require only $O(\log n)$ time to train the neural network; in particular, the sparsified NTK relies on sparse matrix inversion, which is BQP-complete. However, to achieve an exponential speedup in practice over gradient descent, the quantum algorithm for approximating the neural network output must also provide efficient input (i.e. quantum state preparation) and output (i.e. quantum state measurement). Following standard practice in quantum machine learning, we use a QRAM to store the raw dataset; once the QRAM has been prepared, multiple neural networks can be trained with only $O(\log n)$ cost. To efficiently access the trained neural network's output on a test data point $\mathbf x_*$, the quantum state corresponding to the vector $\mathbf k_*$ of the NTK evaluated between the test data and training data must be efficiently preparable, which depends on the data distribution. Similarly, measuring the final inner product is only efficient if there is sufficient state overlap. We provide numerical examples on the MNIST handwritten digit classification dataset to demonstrate that common data distributions satisfy the necessary conditions for efficient input/output, fully realizing the exponential quantum speedup over gradient descent and satisfying the caveats described by~\citet{Aaronson2015}. Beyond the proven fully-connected neural network architecture, numerical experiments based on the non-residual convolutional Myrtle network~\cite{myrtle} provide evidence that the logarithmic quantum training time holds for deep learning architectures similar to those used in real-world settings of classical machine learning. An open source implementation of training a neural network according to the proposed quantum algorithm is provided at~\url{https://github.com/quantummind/quantum-deep-neural-network}.

\section{Neural Tangent Kernel Framework}
We consider the task of binary classification for a dataset $S$ of $n$ training examples $\{(\mathbf x_i, y_i) \in \mathbb{R}^d \times \{-1, 1\}\}_{i=1}^n$. For simplicity, each data example is assumed to be placed on the unit sphere, i.e. $|\mathbf x_i| = 1$. Throughout this work, we will refer to the \emph{separability} of data points $\mathbf x_i, \mathbf x_j$ given by $\delta_{ij} := 1 - |\mathbf x_i \cdot \mathbf x_j|$. Since the depth of the neural network required to provably converge efficiently by gradient descent depends on $\delta := \min_{i,j}\delta_{ij}$, the computational cost of training the neural network is parameterized by the dataset separability.

To classify the dataset, we use the \emph{fully-connected neural network} defined in Eq.~\ref{mt:eq:nn}, consisting of $L$ fully-connected hidden layers of neurons with a nonlinear activation function $\sigma\,:\,\mathbb{R}\to\mathbb{R}$. The neural network output is determined by applying the activation function entry-wise at each neuron to the weights connecting to preceding neurons. We place additional conditions on the normalization of the activation function, which are equivalent to the application of batch normalization at each layer of a neural network~\cite{agarwal2020deep}:
\begin{align}
\underset{X\sim\mathcal{N}(0, 1)}{\mathbb{E}} [\sigma(X)] = 0 \text{  and  } \underset{X\sim\mathcal{N}(0, 1)}{\mathbb{V}} [\sigma(X)] = 1.
\end{align}
Under such normalization, the \emph{coefficient of nonlinearity} $\mu := 1 - \left(\mathbb{E}_{X\sim\mathcal{N}(0,\, 1)}[X \sigma(X)]\right)^2$ defined by~\citet{agarwal2020deep} is bounded by $0 < \mu \leq 1$.

The neural network is trained with $\ell_2$ loss, i.e. gradient descent is used to minimize $\ell(W_1, \dots, W_L, v) = \frac{1}{2}\sum_{i=1}^n (f(\mathbf x_i; W_1, \dots, W_L, v) - y_i)^2$ over the training set. As stated formally in Section~\ref{sm:ntk:conv} of the Supplementary Information and discussed by the recent work of~\citet{agarwal2020deep}, the neural network has a depth $L_\mathrm{conv}$ at which efficient convergence by gradient descent is guaranteed:
\begin{align}
    \label{mt:eq:lconv}
    L_\mathrm{conv} &:= \frac{8\log(n/\delta)}{\mu}.
\end{align}
For depth $L \geq L_\mathrm{conv}$, gradient descent will find an $\epsilon$-suboptimal point in $O(\log(1/\epsilon))$ iterations with high probability for standard choices of activation function.

The NTK defines a kernel function between two data examples and may be written explicitly for the neural network of Eq.~\ref{mt:eq:nn} in terms of the dual activation function and its derivatives (see Section~\ref{sm:ntk:comp} of the Supplementary Information). Under our dataset framework, it simplifies to a function of only the inner product $\mathbf x_i \cdot \mathbf x_j$. Evaluated over the entire training set, the $(i, j)$th element of the $n \times n$ matrix $K$ is defined by the NTK $k(\mathbf x_i, \mathbf x_j)$ between the respective data points. The cost of evaluating the NTK between a pair of data points scales polynomially in the neural network depth $L$. Since the depth $L_\mathrm{conv}$ necessary for efficient gradient descent scales like $O(\log (n/\delta))$, we require that $\delta = \Omega(1/\mathrm{poly}\;n)$ to compute any given element of the NTK matrix in $O(\log n)$ time. In the Supplementary Information (Section~\ref{sm:data}), this scaling of $\delta$ is verified numerically for the MNIST dataset and is further theoretically motivated by considering a uniform distribution on the unit sphere.

For a neural network of depth $L \geq L_\mathrm{conv}$, the matrix $K$ becomes well-conditioned: the condition number $\kappa(K)$ (i.e. ratio of the largest to smallest singular value) is bounded by
\begin{align}
    1 \leq \kappa(K) \leq \frac{1 + 1/n}{1 - 1/n},
\end{align}
allowing gradient descent to efficiently converge as $n$ increases. Such well-conditioning is due to the structure of the NTK for a deep neural network. Under the above data assumptions, all diagonal matrix elements of $K$ are equal and larger than all off-diagonal elements, which are bounded by $\left|\frac{K_{ij}}{K_{11}}\right| \leq \left(\frac{\delta}{\delta_{ij} n}\right)^2$ if $0 < \delta_{ij} < 1/2$, and $\left|\frac{K_{ij}}{K_{11}}\right| \leq \left(\frac{\delta}{n}\right)^2$ if $1/2 \leq \delta_{ij} \leq 1$ (see Section~\ref{sm:ntk:bounds} of the Supplementary Information). Hence, as the training set size increases, the off-diagonal matrix elements of $K$ vanish, with the largest off-diagonal elements corresponding to the most similar data examples. We use this characteristic of the NTK to motivate the following two approximations of the matrix $K$.

The \emph{diagonal} NTK is defined by setting all off-diagonal elements of the NTK to zero. The neural network output for a test data example $\mathbf x_*$ is then directly proportional to the inner product $\mathbf{k}_*^T \mathbf y$, allowing the matrix inversion to be omitted entirely. The \emph{sparsified} NTK is defined by only permitting $O(\log n)$ off-diagonal elements to be nonzero in any row or column. As shown in the Supplementary Information (Section~\ref{sm:approx}), the error of both approximations is bounded by $O(1/n)$, ensuring convergence to the exact neural network output for large training sets. However, the bound of the sparsified NTK given by the Gershgorin circle theorem is strictly tighter, suggesting superior performance compared to the diagonal NTK. In the Numerical Experiments section, we confirm this for the MNIST benchmark.

\section{Quantum algorithm}
To train a wide neural network, the NTK formalism requires inversion of the $n \times n$ matrix $K$ representing the NTK evaluated over the entire training set (Eq.~\ref{mt:eq:f-ntk}): under the approximations described above, we instead seek to evaluate $\mathbf{k}_*^T \mathbf y$ or $\mathbf{k}_*^T \tilde{K}^{-1}\mathbf y$ for a sparse matrix $\tilde K$ corresponding to the sparsified approximation to the NTK. Hence, a linear equation and inner product must be evaluated to approximate the neural network output. The HHL algorithm proposed by Harrow, Hassidim and Lloyd~\cite{PhysRevLett.103.150502} and subsequent improvements~\cite{doi:10.1137/16M1087072} solve the \emph{quantum linear systems problem} (QLSP) of $A \ket{x} = \ket{b}$ corresponding to the linear equation $A\mathbf x = \mathbf b$, given access to a procedure $\mathcal{P}_A$ that computes the locations and values of the nonzero entries in $A$ and a procedure $\mathcal{P}_B$ that prepares the state $\ket{b}$. For a matrix $A$ with condition number $\kappa$ and at most $s$ nonzero entries per row or column, solving the QLSP up to error $\epsilon$ requires $O(\log(n) \kappa s \,\mathrm{polylog}(\kappa s/\epsilon))$ time~\cite{doi:10.1137/16M1087072}.

To achieve an exponential quantum speedup with a quantum linear systems algorithm (QLSA), it is evident that $s$ and $\kappa$ must scale like $O(\mathrm{polylog}(n))$; similarly, the procedures $\mathcal{P}_A$ and $\mathcal{P}_B$ must take $O(\mathrm{polylog}(n))$ time. Due to dequantization results for low-rank matrices~\cite{10.1145/3313276.3316310}, the matrix $A$ must also have rank at least $\Omega(\mathrm{poly}(n))$. Finally, in order to efficiently access the resulting state $\ket{x}$ after solving $A\ket{x} = \ket{b}$, the readout must require at most $O(\mathrm{polylog}(n))$ measurements.

The full quantum algorithm for the approximate NTK is summarized in Algorithm~\ref{alg}. We use the neural network defined in Eq.~\ref{mt:eq:nn}, although we discuss extensions to other architectures in the Supplementary Information (Section~\ref{sm:qntk:cnn}) and provide additional analyses in the Numerical Experiments section.

As is typical in many quantum machine learning algorithms~\cite{lloyd2013quantum,Biamonte2017,PhysRevLett.113.130503,PhysRevA.94.022342,NEURIPS2019_16026d60}, we assume the existence of a quantum random access memory (QRAM) to store and access any necessary quantum states. A binary tree QRAM data structure~\cite{qram} may be applied similarly to existing work in quantum machine learning~\cite{NEURIPS2019_16026d60,Kerenidis2020Quantum}. For any data analysis application, writing a dataset requires $O(n)$ time; however, the cost only occurs once. In the case of a quantum NTK algorithm, multiple neural networks may be trained after the data is stored in QRAM.

To compute the NTK elements, we can explicitly evaluate inner products between data examples using amplitude estimation~\cite{brassard2002quantum} and median evaluation~\cite{wiebe2014quantum}. While the NTK of the fully-connected neural network defined by Eq.~\ref{mt:eq:nn} only depends on the inner product $\mathbf x_i \cdot \mathbf x_j$, similar techniques can be extended to convolutional and pooling layers with additional qubits and computation independent of $n$. Preparation of the NTK between the test data point $\mathbf x_*$ and the training set $S$ requires post-selection, the cost of which depends on the dataset. In the Numerical Experiments section, we show that this requires $O(\log n)$ time for classification of MNIST handwritten digits.

\begin{algorithm*}[t]
{\small
\begin{algorithmic}[1]
\caption{{\small Training a wide and deep fully-connected feedforward neural network.}}
\label{alg}

\Statex
\Function{PrepareTestKernel}{$\mathbf x_*, S$}
\Comment Assumes dataset $S = \{(\mathbf x_i, y_i)\}$ has already been loaded into QRAM

\State Apply amplitude estimation and median evaluation to prepare the state $\frac{1}{\sqrt{n}}\sum_{i=1}^{n} \ket{i}\ket{\mathbf x_* \cdot \mathbf x_i}$

\State Apply the NTK over the superposition to prepare $\frac{1}{\sqrt{n}}\sum_{i=1}^{n} \ket{i}\ket{k(\mathbf x_*, \mathbf x_i)}$
\LeftComment{NTK is only a function of the inner product}

\State Post-select with $O(1/P)$ measurements to prepare $\ket{k_*} = \frac{1}{\sqrt{P}} \sum_{i=1}^n k(\mathbf x_*, \mathbf x_i)\ket{i}$ 
\LeftComment{$P$ depends on data distribution}

\State Output $\ket{k_*}$
\EndFunction

\Statex
\Function{DiagonalNTK}{$\mathbf x_*, S$}

\State Prepare $\ket{k_*}$ from \Call{PrepareTestKernel}{$\mathbf x_*, S$}

\State Load $\ket{y} = \frac{1}{\sqrt{n}}\sum_{i=1}^n y_i\ket{i}$ from QRAM

\State Introduce ancilla qubit to encode relative phase from $\frac{1}{2}(\ket{0}(\ket{k_*} + \ket{y}) + \ket{1}(\ket{k_*}-\ket{y}))$
\State Output $\mathrm{sign}(\bra{k_*}\ket{y})$ after $O(1/|\bra{k_*}\ket{y}|^2)$ measurements
\Comment $|\bra{k_*}\ket{y}|^2$ depends on data distribution

\EndFunction

\Statex
\Function{SparsifiedNTK}{$\mathbf x_*, S, \nu$}
\Comment Deterministic symmetric sparsity map $\nu:[N]\times [s] \to [N]$

\State Prepare $\ket{k_*}$ from \Call{PrepareTestKernel}{$\mathbf x_*, S$}

\State Define $\mathcal{P}_A$ that computes nonzero matrix elements of sparsified NTK $\tilde K$ from sparsity map $\nu$ and NTK $k(\mathbf x_i, \mathbf x_j)$

\State Define $\mathcal{P}_B$ as loading $\ket{y} = \frac{1}{\sqrt{n}}\sum_{i=1}^n y_i\ket{i}$ from QRAM

\State Solve $\tilde K \ket{v} = \ket{y}$ with QLSA given $\mathcal{P}_A, \mathcal{P}_B$
\Comment $O(\mathrm{polylog}\; n)$ time since $\kappa, s = O(\mathrm{polylog}\; n)$

\State Introduce ancilla qubit to encode relative phase from $\frac{1}{2}(\ket{0}(\ket{k_*} + \ket{v}) + \ket{1}(\ket{k_*}-\ket{v}))$
\State Output $\mathrm{sign}(\bra{k_*}\ket{v})$ after $O(1/|\bra{k_*}\ket{v}|^2)$ measurements
\Comment $|\bra{k_*}\ket{v}|^2$ depends on data distribution

\EndFunction
\end{algorithmic}
}
\end{algorithm*}

For the sparsified NTK, a deterministic sparsification pattern must be selected such that at most $O(\log n)$ elements are nonzero in any row or column. Since the NTK approaches an identity matrix, the chief constraint on the sparsity pattern is that the diagonal is fixed to be nonzero; the remaining elements may be chosen pseudorandomly. Since the matrix is required to have sparsity $s = O(\log n)$ independent of the type of neural network, the sparsity pattern containing the locations of $O(n \log n)$ nonzero matrix elements may be created once when the dataset is initially stored in QRAM and thereafter be used by any neural network. Since $s = O(\log n)$ and the Gershgorin circle theorem ensures $\kappa$ approaches $1$ as $n$ increases, the QLSP associated with the sparsified NTK is efficient to solve. Moreover, since the NTK is full-rank, the quantum algorithm cannot be dequantized by recent results for efficiently solving low-rank linear systems with classical algorithms~\cite{10.1145/3313276.3316310}.

Finally, the inner product $\bra{k_*}\ket{y}$ (for the diagonal NTK) or $\bra{k_*}\tilde K^{-1}\ket{y}$ (for the sparsified NTK) must be evaluated. Since the task is binary classification, we measure the sign via an inner product estimation subroutine~\cite{zhao2019compiling}. The state overlap must be sufficient to ensure that only $O(\log n)$ measurements are required. Much like the post-selection in preparing $\ket{k_*}$, this depends on the data distribution and is numerically shown to be efficient for an image classification problem, enabling an exponential speedup over gradient descent.

\section{Numerical experiments}

To support the theoretical result above and extend it beyond the proven regime, we consider the MNIST binary image classification task between pairs of digits. Two infinite-width neural network architectures are used: a feedforward fully-connected neural network (Eq.~\ref{mt:eq:nn}) and an architecture based on the convolutional Myrtle network~\cite{myrtle}. In each case, the handwritten image dataset is projected onto the surface of a unit sphere, which may be done during the initial encoding into QRAM. In the case of the sparsified NTK approximation, we enforce sparsity $s = O(\log n)$ for the QLSP associated with the NTK by deterministically choosing a sparsity pattern with a pseudorandom generator. To evaluate if the NTK approximations to training a wide and deep neural network require $O(\log n)$ time, we must numerically check if the dataset-dependent quantities --- i.e., the number of measurements required for post-selection and readout --- scale logarithmically with training set size. Since the approximation introduces $O(1/n)$ error in estimating the trained neural network output, the accuracies of the approximate neural networks are compared to that of an exact wide and deep neural network. All experiments are implemented with the \texttt{neural-tangents} package~\cite{neuraltangents2020}.

\subsection{Fully-connected neural network}

Following the architecture described in Eq.~\ref{mt:eq:nn}, we prepare an infinite-width feedforward fully-connected neural network with normalized erf activation functions, which uniquely specifies the coefficient of nonlinearity $\mu$. To provably converge efficiently by gradient descent, the neural network requires depth $L \geq L_\mathrm{conv} = \frac{8\log(n/\delta)}{\mu}$. As shown in the Supplementary Information (Section~\ref{sm:data}), the dataset separability for the MNIST dataset is bounded by $\delta = \Omega(1/\mathrm{poly}\;n)$, ensuring that any given matrix element of the NTK can be computed in $O(\log n)$ time. To reduce the computational burden, we evaluate the NTK at a neural network depth of $L_\mathrm{conv}/10$; for a subset of the MNIST dataset of size up to $n = 512$, this corresponds to neural networks with fewer than 100 hidden layers. Nevertheless, the depth is sufficient for the sparsified NTK to have a monotonically decreasing condition number at $\kappa \approx 1$, satisfying the condition $\kappa = O(\log n)$ required for an efficient quantum linear systems algorithm. With both $\kappa$ and $s$ bounded by $O(\log n)$, the QLSP of the sparsified NTK may be solved in logarithmic time with respect to the training set size.

To classify a test data point $\mathbf x_*$, we also require the NTK to be evaluated between $\mathbf x_*$ and the training set, i.e. the state $\ket{k_*} = \frac{1}{\sqrt{P}} \sum_{i=1}^n k(\mathbf x_*, \mathbf x_i)\ket{i}$ must be efficiently prepared. This requires post-selection to perform an amplitude encoding, requiring $O(1/P)$ measurements for normalization factor $P = \sum_{i=1}^n k^2(\mathbf x_*, \mathbf x_i)$. As shown in Figure~\ref{fig:ff:ps}, the quantity $1/P$ scales like $O(\log n)$, allowing $\ket{k_*}$ to be prepared efficiently. Once $\ket{k_*}$ is prepared, the classification of test data point $\mathbf x_*$ by the neural network is determined by the sign of the inner product $\bra{k_*}\ket{y}$ (for a diagonal approximation) or $\bra{k_*}\tilde K^{-1}\ket{y}$ (for a sparsified approximation). Measuring the sign of the inner product requires sufficient overlap between $\ket{k_*}$ and either $\ket{y}$ or $\tilde K^{-1}\ket{y}$, with the number of measurements scaling like the reciprocal of the overlap squared. For the MNIST binary classification between digits 8 and 9, we find that this is bounded by $O(\log n)$ (Figure~\ref{fig:ff:read}), completing the final requirement for an end-to-end runtime of $O(\mathrm{polylog}\; n)$ to approximately train the neural network and apply it to test data.

\begin{figure}[H]
  \centering
  \begin{subfigure}[t]{0.45\textwidth}
  \centering
  \includegraphics[width=0.73\textwidth]{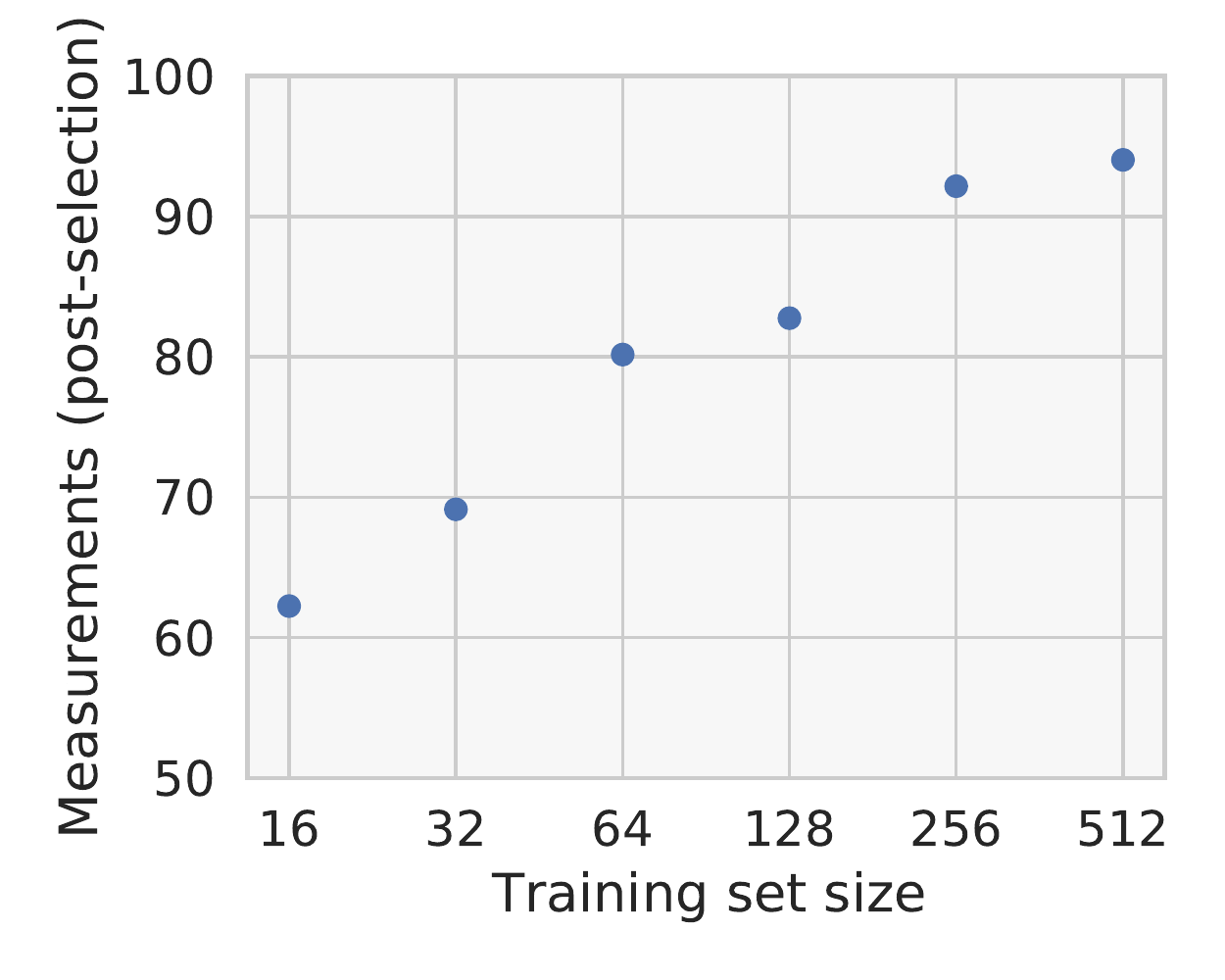}
  \caption{Scaling of post-selection to prepare $\ket{k_*}$\label{fig:ff:ps}}
  \end{subfigure}%
  \hspace{0.05\textwidth}%
  \begin{subfigure}[t]{0.45\textwidth}
  \centering
  \includegraphics[width=0.73\textwidth]{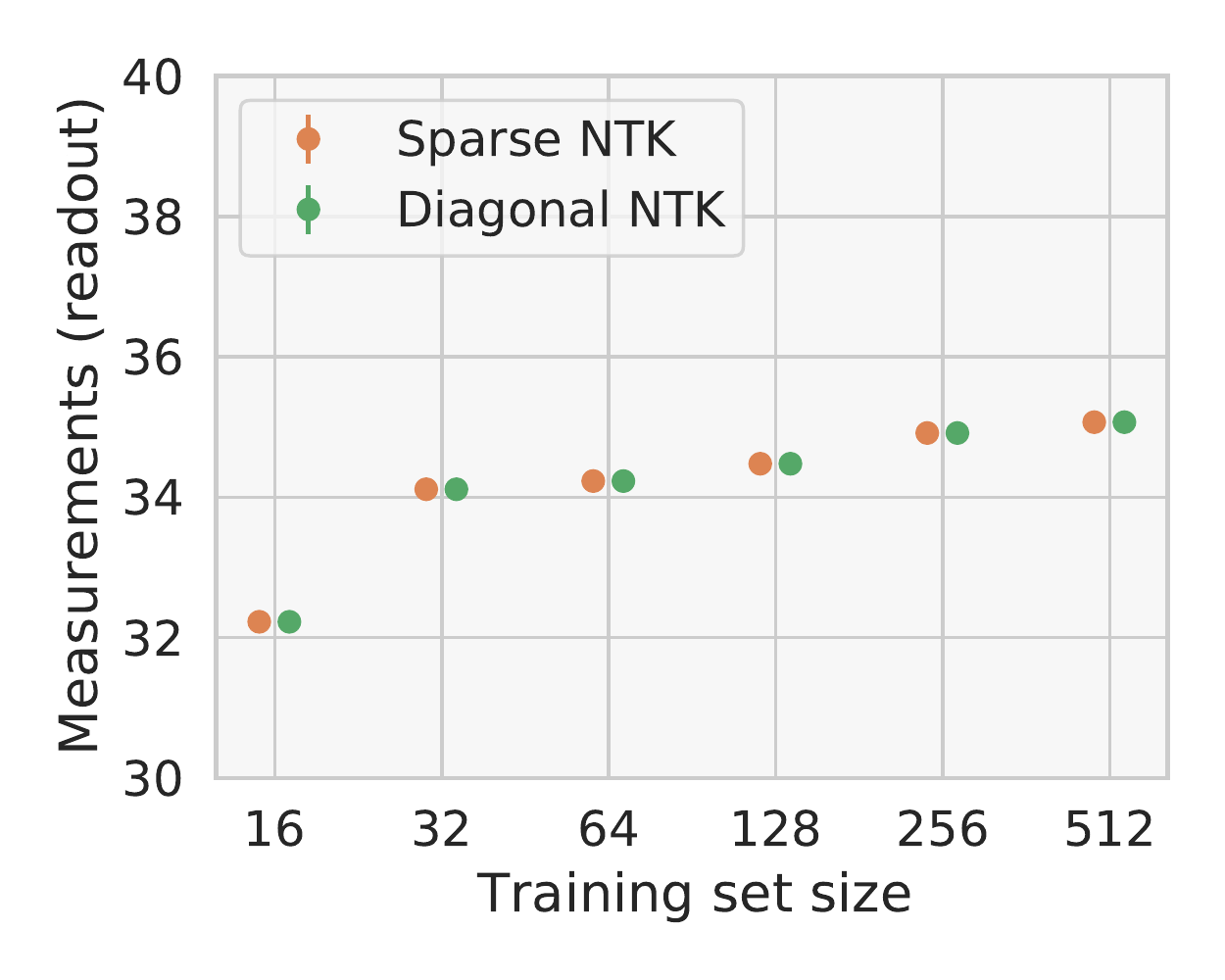}
  \caption{Scaling of readout for binary classification\label{fig:ff:read}}
  \end{subfigure}
  \caption{\textit{Data-dependent efficiency of the NTK for MNIST 8 vs. 9 classification.} \textbf{(a)} Median number of measurements for post-selection scales like $O(\log n)$. \textbf{(b)} Median number of measurements for the final sign estimation is upper-bounded by $O(\log n)$. Error bars (too small to be visible) show two standard deviations estimated by Poisson bootstrapping.}
\end{figure}

However, since the NTK of the fully-connected neural network rapidly approaches a diagonal matrix (Figure~\ref{fig:ff:ntk}), no difference is apparent in the classification performance of the exact, sparsified, or diagonal NTK (Figure~\ref{fig:ff:acc}). In general, the sparsified NTK may be expected to have lower error than the diagonal NTK in estimating the exact neural network output, due to the tighter error bound provided by the Gershgorin circle theorem (see Section~\ref{sm:approx} of the Supplementary Information). Although the performance is identical in this example, the following experiment with a convolutional neural network demonstrates the improved performance of the sparsified NTK.

\begin{figure}[H]
  \centering
  \begin{subfigure}[t]{0.5\textwidth}
  \centering
  \hspace*{0.5em}
  \raisebox{.08\height}{%
  \includegraphics[width=0.73\textwidth]{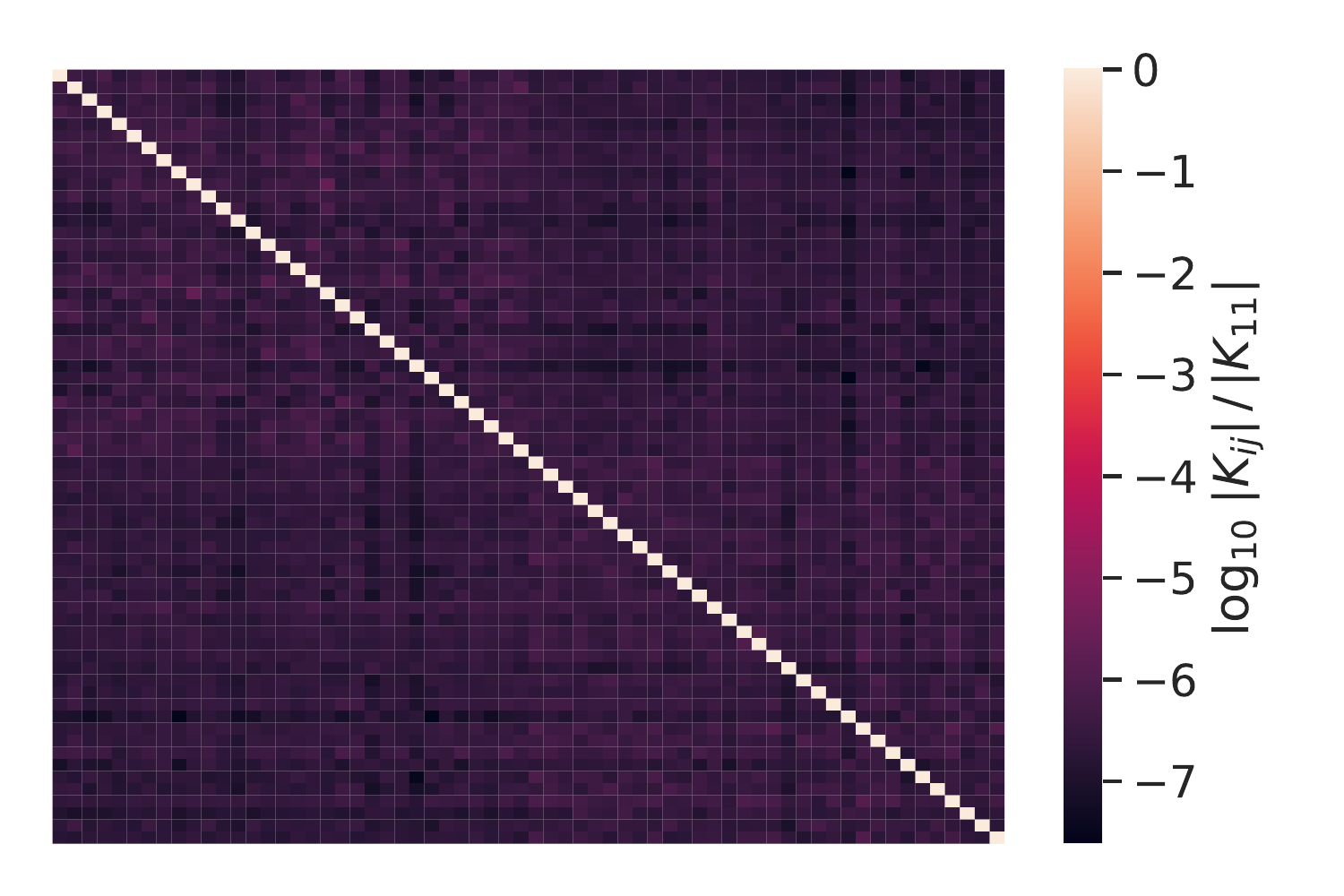}}%
  \caption{Neural tangent kernel for $L = L_\mathrm{conv}/10$ ($n=64$)\label{fig:ff:ntk}}
  \end{subfigure}%
  ~
  \begin{subfigure}[t]{0.45\textwidth}
  \centering
  \includegraphics[width=0.73\textwidth]{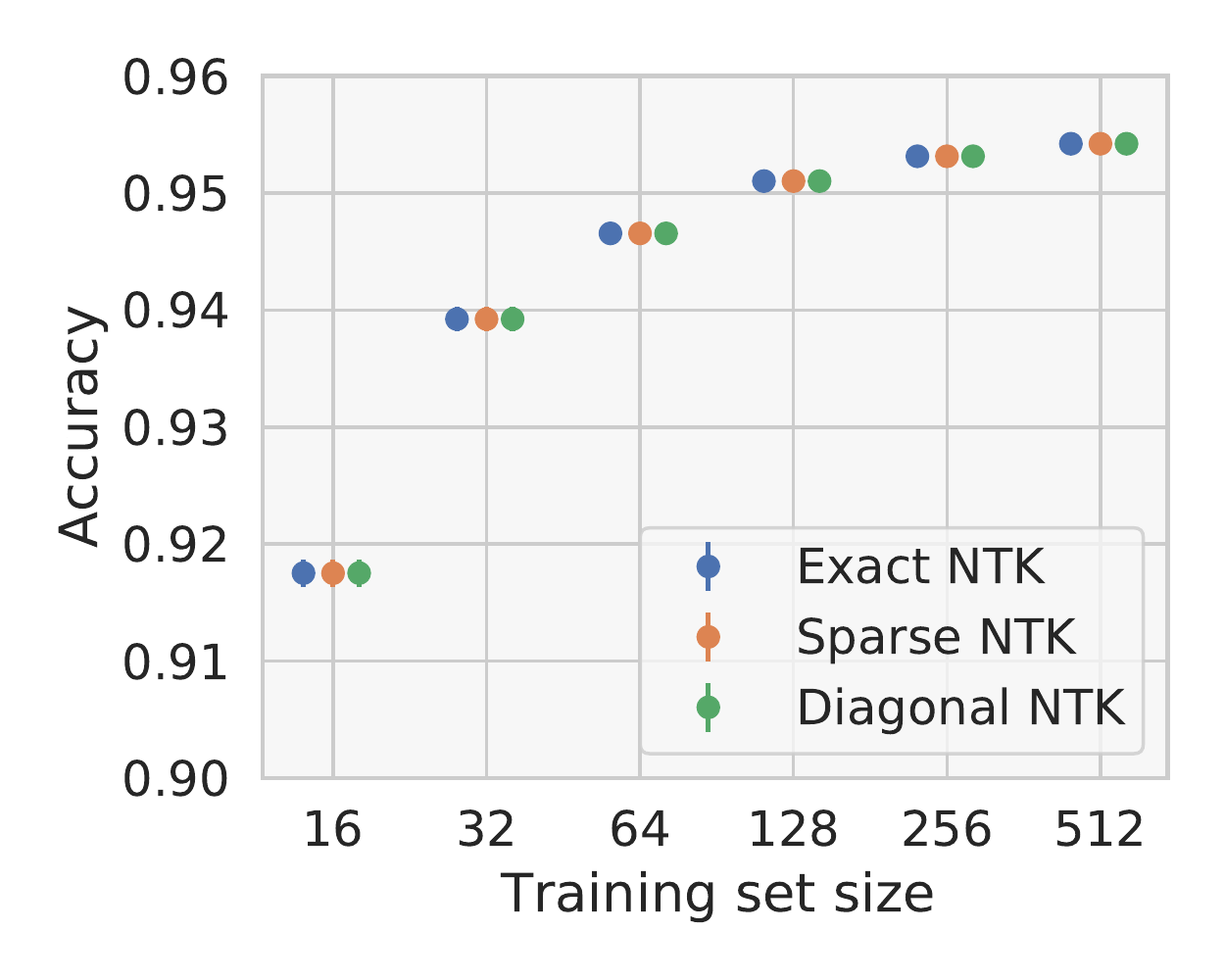}
  \caption{Accuracy (test set)\label{fig:ff:acc}}
  \end{subfigure}
  \caption{\textit{Neural tangent kernel performance for MNIST 8 vs. 9 classification.} \textbf{(a)} Matrix elements of the NTK over the training set ($n=64)$, demonstrating the vanishing of off-diagonal elements. \textbf{(b)} Accuracy of the exact and approximate NTKs on the test dataset. Error bars show two standard deviations.}
\end{figure}

\subsection{Convolutional neural network}

We choose the Myrtle network~\cite{myrtle} due to its straightforward architecture and use in previous benchmarks~\cite{pmlr-v119-shankar20a,lee2020finite}. Consisting of $3 \times 3$ convolutional layers with ReLU activation functions, strided average pooling, and a final fully connected layer, the Myrtle network resembles common non-residual convolutional neural networks used in classical machine learning. As described in the Supplementary Information (Section~\ref{sm:num:cnn}, the Myrtle NTK is also observed to be well-conditioned with vanishing off-diagonal elements as its depth increases, similarly to the proven case of the fully-connected neural network.

Motivated by the widespread use of standard-size architectures (e.g. ResNet-152~\cite{He_2016_CVPR}), we push further beyond the proven regime and examine the behavior of a neural network with a fixed depth of $L = 101$. Although a depth of $L_\mathrm{conv}$ that increases with $n$ improves conditioning as the training set increases in size, the choice of a fixed depth requires the introduction of matrix preconditioning. The naturally emphasized diagonal of the convolutional NTK --- similarly to the proven structure of the fully-connected network --- enables the sparsified NTK $\tilde K$ to be effectively preconditioned by further exaggerating the vanishing of off-diagonal elements. Assuming training set examples are drawn i.i.d. from the same underlying data distribution, increasing $n$ has a consistent effect on $\kappa$ of the sparsified NTK: as expected by the Gershgorin circle theorem applied to $O(\log n)$ off-diagonal elements, the preconditioning requires off-diagonal elements to be decreased by a multiplicative factor that scales at most logarithmically with the training set size. Having efficiently conditioned $\tilde K$ such that $\kappa = O(\log n)$ and chosen a deterministic sparsity pattern that ensures $s = O(\log n)$, the runtime to solve the QLSP is bounded to be polylogarithmic in training set size. The exact NTK, sparsified NTK, and diagonal NTK are shown in Figure~\ref{fig:my:ntk}. Further discussion of the conditioning of the convolutional NTK is provided in the Supplementary Information (Section~\ref{sm:num:cnn}).

\begin{figure}[H]
  \centering
  \begin{subfigure}[t]{0.33\textwidth}
  \includegraphics[width=\textwidth]{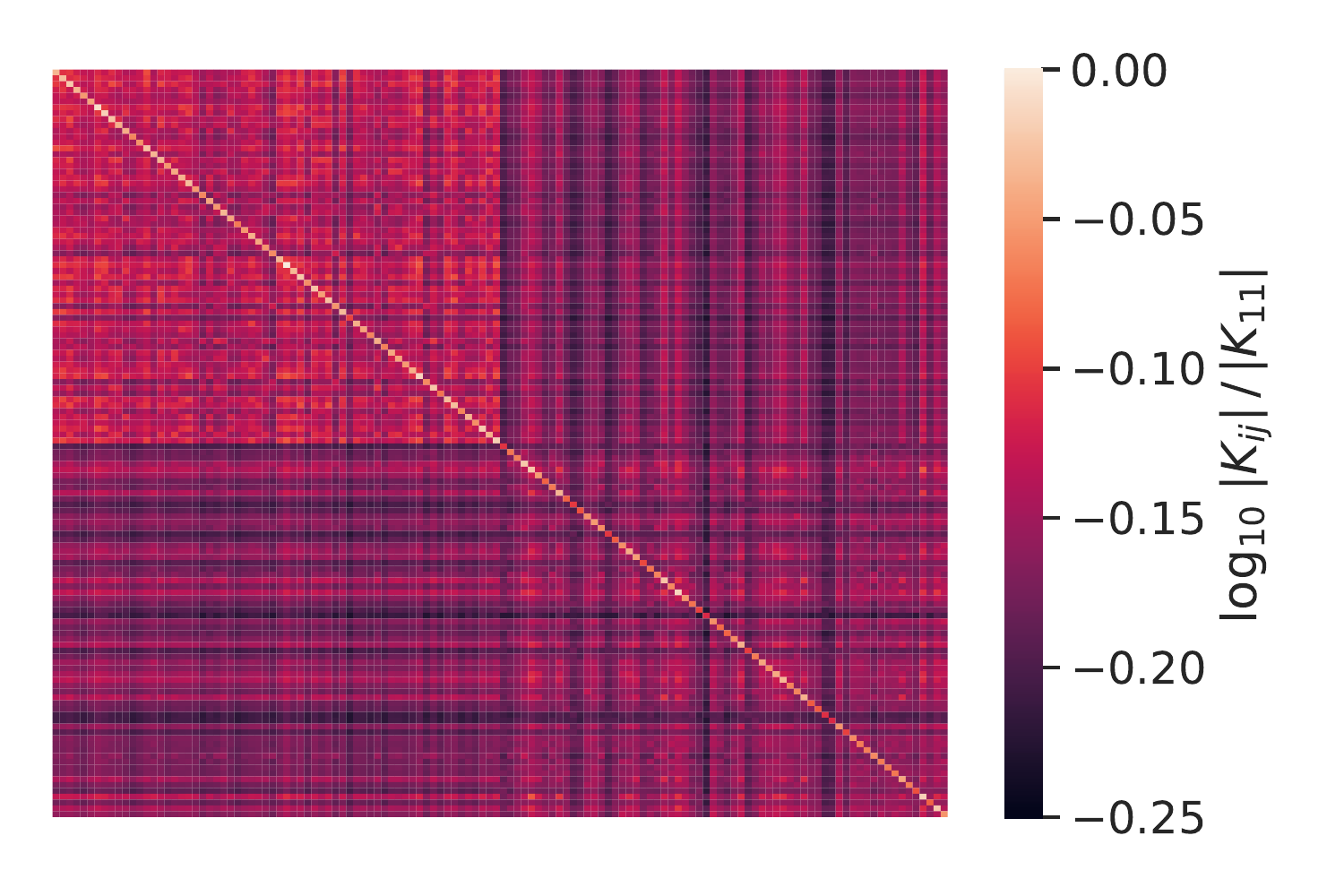}
  \caption{Exact neural tangent kernel}
  \end{subfigure}%
  \begin{subfigure}[t]{0.33\textwidth}
  \includegraphics[width=\textwidth]{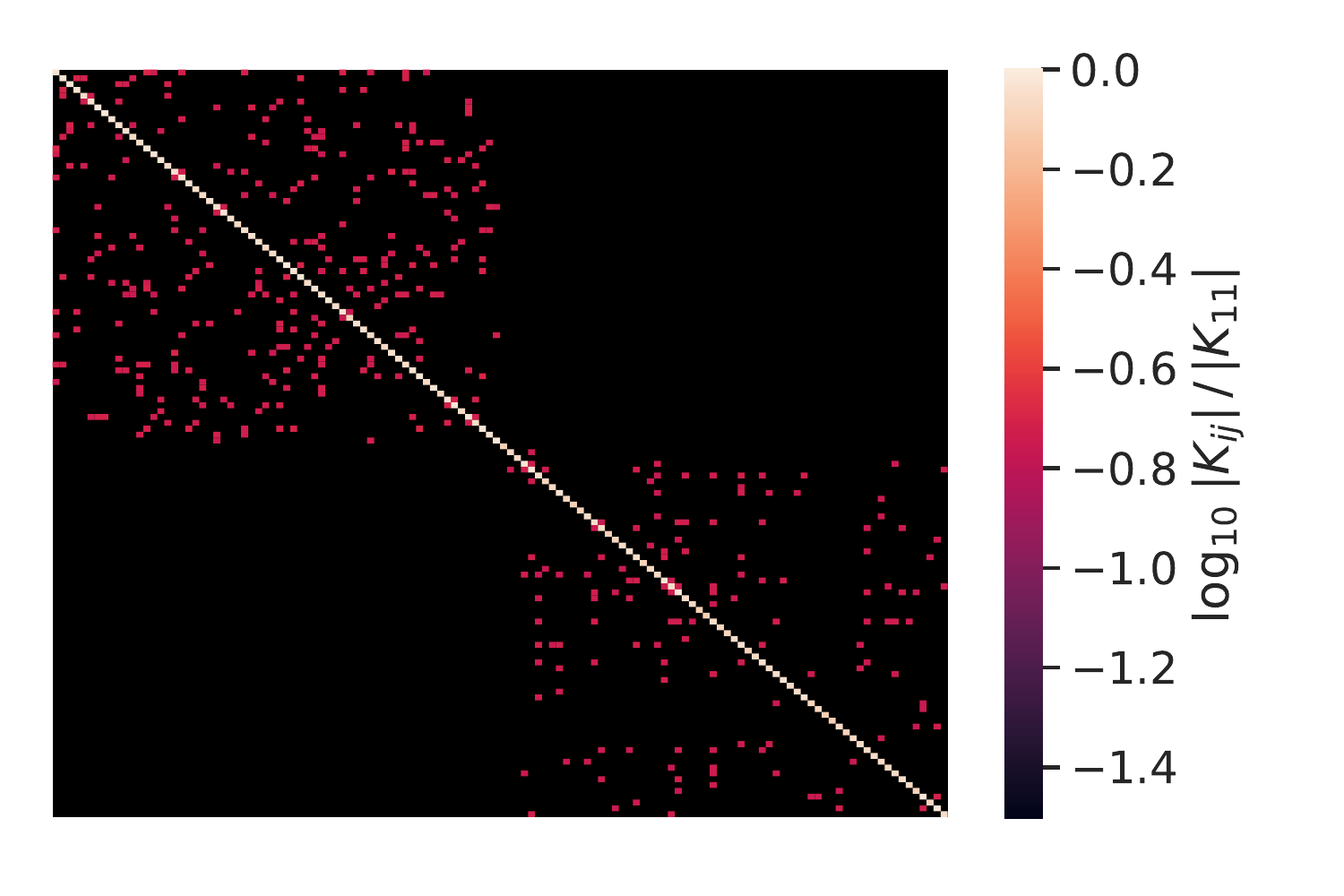}
  \caption{Sparsified neural tangent kernel}
  \end{subfigure}%
  \begin{subfigure}[t]{0.33\textwidth}
  \includegraphics[width=\textwidth]{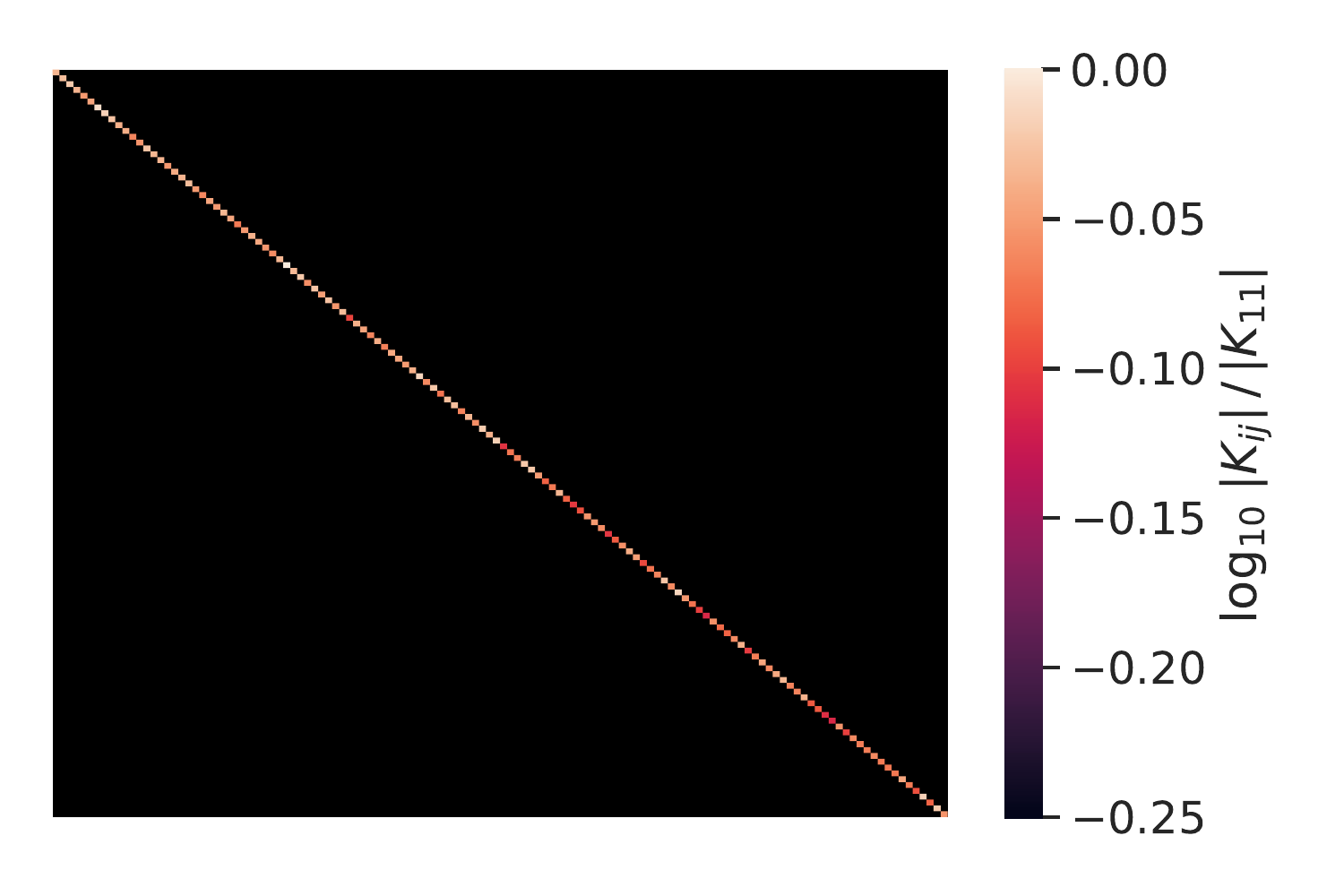}
  \caption{Diagonal neural tangent kernel}
  \end{subfigure}
  \caption{\textit{Exact and approximate NTKs of the convolutional neural network for MNIST 0 vs. 1 classification ($n=128$).} \textbf{(a)} The exact convolutional NTK exhibits similar well-conditioning to the fully-connected neural network (Figure~\ref{fig:ff:ntk}). \textbf{(b)} The sparsified NTK is generated by a fixed pseudorandom sparsity pattern and preconditioning applied by reducing off-diagonal elements by a multiplicative factor (extrapolated using the Gershgorin circle theorem). \textbf{(c)} The diagonal convolutional NTK has distinct entries (unlike the fully-connected NTK) due to convolutional operations on the data. Note that the NTK is shown over the sorted training set; elements between the same class appear as block diagonal matrices.}
  \label{fig:my:ntk}
\end{figure}

Since the convolutional neural network depth is not sufficiently large for the NTK to be well-approximated by a diagonal matrix, we observe improved performance of the sparsified NTK compared to the diagonal NTK. Similarly to the fully-connected neural network, the number of measurements required to post-select the state $\ket{k_*}$ and to determine the final classification are bounded by $O(\log n)$. Since state preparation and readout are efficient, and since the QLSP is provided a well-conditioned and sparse matrix with efficiently preparable entries, the numerical evidence indicates that the runtime is polylogarithmic in the training set size (Figure~\ref{fig:my:perf}).

\begin{figure}[H]
  \centering
  \begin{subfigure}[t]{0.33\textwidth}
  \includegraphics[width=\textwidth]{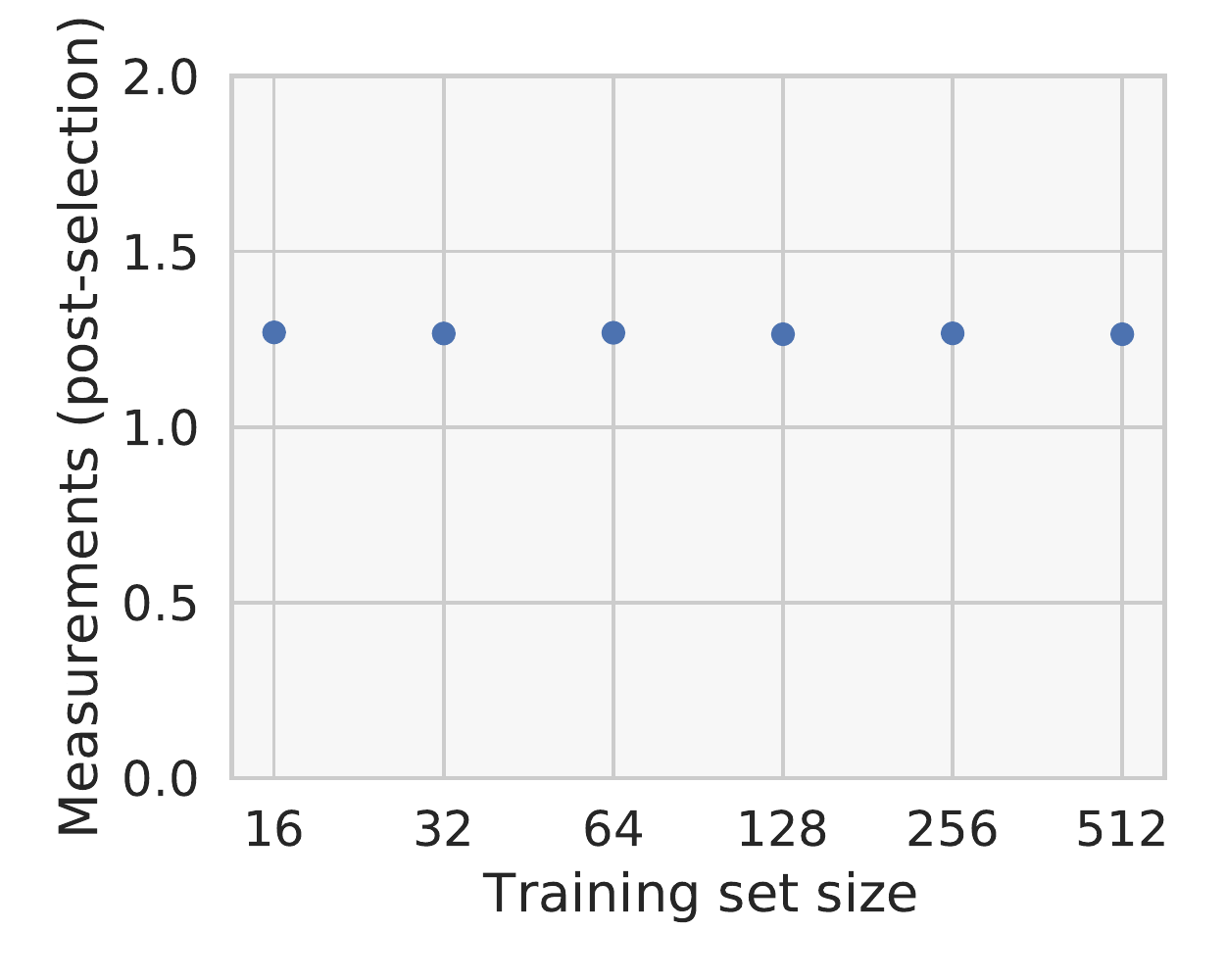}
  \caption{Post-selection to prepare $\ket{k_*}$}
  \end{subfigure}%
  \begin{subfigure}[t]{0.33\textwidth}
  \includegraphics[width=\textwidth]{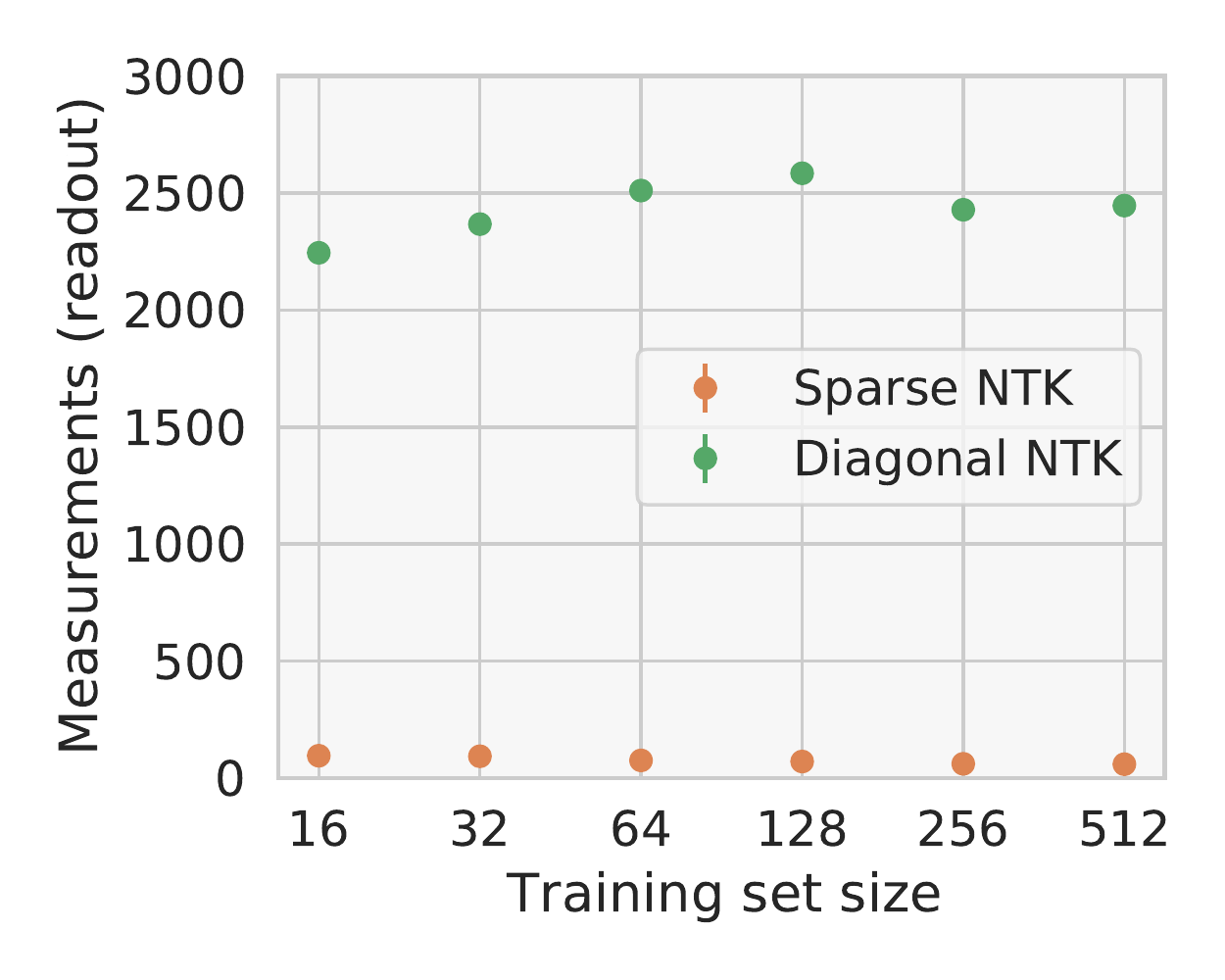}
  \caption{Readout for binary classification}
  \end{subfigure}%
  \begin{subfigure}[t]{0.33\textwidth}
  \includegraphics[width=\textwidth]{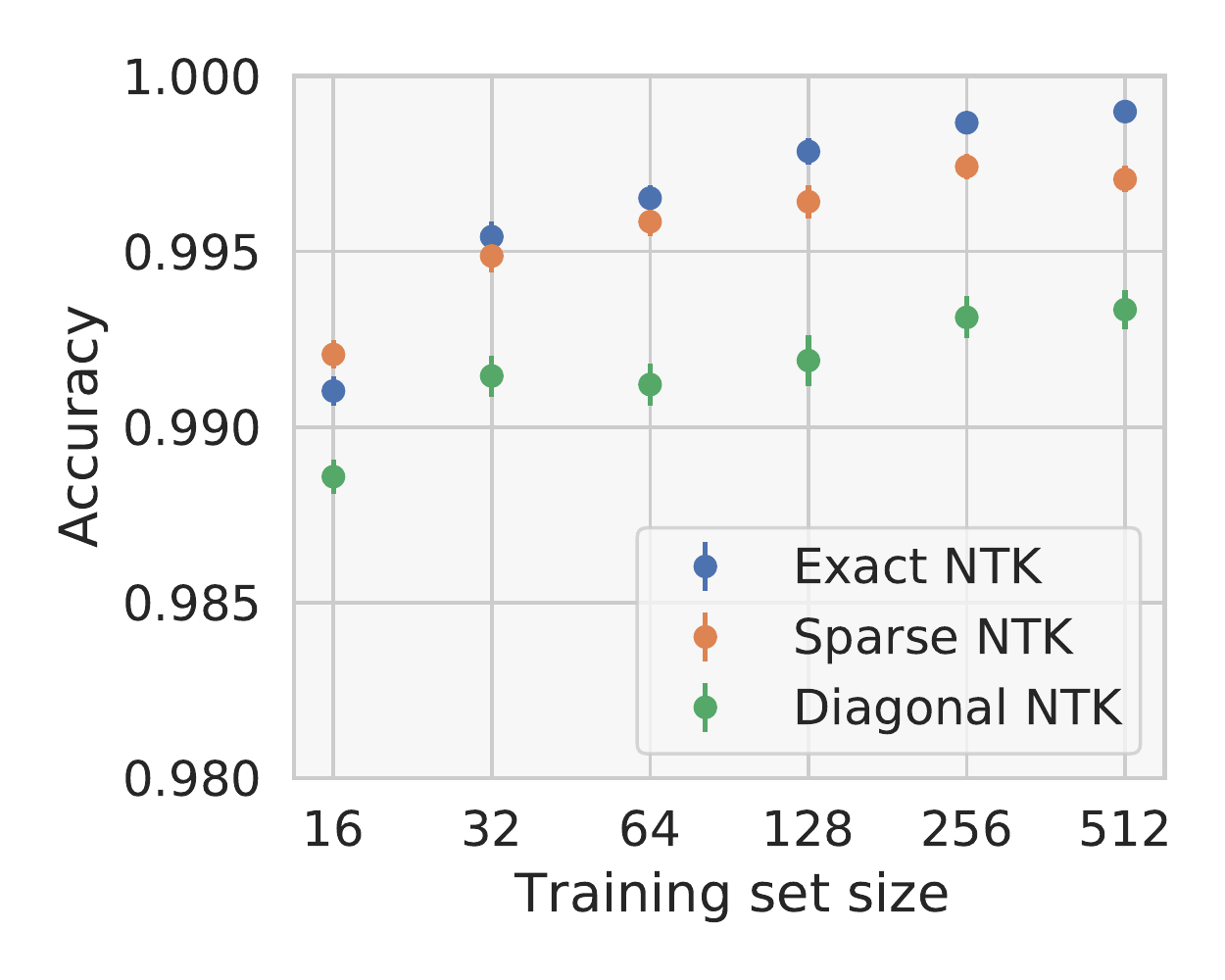}
  \caption{Accuracy (test set)}
  \end{subfigure}
  \caption{\textit{Scaling and performance of the convolutional neural network for MNIST 0 vs. 1 classification.} \textbf{(a)} The number of measurements required for post-selection and \textbf{(b)} readout are constant due to the fixed neural network depth $L = 101$. \textbf{(c)} Although both approximations approach the exact neural network output, the sparsified NTK achieves better performance than the diagonal NTK.}
  \label{fig:my:perf}
\end{figure}

\section{Outlook}

Motivated by the success of classical deep learning, this work demonstrates a quantum algorithm to train classical neural networks in logarithmic time and provides numerical evidence of such efficiency on the standard MNIST image dataset. Beyond the proven case of the feedforward fully-connected neural network, the neural tangent kernel offers a versatile framework across architectures such as convolutional neural networks, graph neural networks, and transformers~\cite{NIPS2014_81ca0262,pmlr-v119-shankar20a,NEURIPS2019_663fd3c5,yang2020tensor}. Hence, we anticipate the applicability of our approach to a diverse set of neural networks, consistent with the presented empirical results for a convolutional architecture with pooling layers and ReLU activation functions. While our work focused on the neural tangent kernels corresponding to neural networks, the approximation introduced by a sparsified or diagonal kernel may extend to any \emph{chaotic kernel}~\cite{NIPS2016_14851003,46760}. As the depth of a chaotic kernel increases, similar data entries become increasingly dissimilar due to random projections onto weight matrices; this may generally give rise to the kernel structure key to well-conditioning and successful approximation in logarithmic time. Given the interest within the quantum machine learning community on kernel approaches due to the exponentially large Hilbert space offered by quantum computing~\cite{Biamonte2017,PhysRevLett.122.040504,Havlicek2019,Blank2020,lloyd2020quantum,schuld2021quantum}, this work may open new possibilities for improved machine learning methods amenable to quantum computing beyond the scope of classical neural networks.

\section{Data availability}
The MNIST dataset used for experiments is publicly available~\cite{726791}. Data supporting plots are available from the corresponding author upon reasonable request.

\section{Code availability}
Source code implementing the proposed quantum algorithm is available at~\url{https://github.com/quantummind/quantum-deep-neural-network}. 

\begin{acknowledgements}
The authors thank Naman Agarwal and Patrick Rall for technical advice, and Maria Spiropulu for facilitating discussions. AZ acknowledges support from Caltech's Intelligent Quantum Networks and Technologies (INQNET) research program and by the DOE/HEP QuantISED program grant, Quantum Machine Learning and Quantum Computation Frameworks (QMLQCF) for HEP, award number DE-SC0019227.
\end{acknowledgements}

\bibliography{scibib}

\newpage
\begin{center}
\textbf{\Large Supplementary information}
\end{center}

\renewcommand{\thefigure}{S\arabic{figure}}
\setcounter{figure}{0}
\setcounter{equation}{0}
\setcounter{theorem}{0}
\setcounter{secnumdepth}{3}

\section{Properties of the neural tangent kernel}
\label{sm:ntk}

Before providing proofs of the central results presented in the main text, we provide preliminary definitions and results related to the neural tangent kernel (NTK) that are used throughout the Supplementary Information.

\subsection{Preliminaries}
To provide the necessary notation for the NTK, we formalize the definitions and data assumptions of the main text. Consider a binary classification dataset $S$ of $n$ training examples $\{(\mathbf x_i, y_i) \in \mathbb{R}^d \times \{-1, 1\}\}_{i=1}^n$. To parameterize our results, we must define the separability between data examples.

\begin{definition}[Separability]
The separability of data points $\mathbf x_i, \mathbf x_j$ is given by $\delta_{ij} := 1 - |\mathbf x_i \cdot \mathbf x_j|$. The separability of the dataset is $\delta := \min_{i, j}\delta_{ij}$.
\end{definition}

We make the following standard assumption about separability across the entire dataset~\cite{NIPS2018_8038,pmlr-v97-allen-zhu19a,NEURIPS2019_6a61d423} with an additional lower bound on the separability that is commonly satisfied (see Section~\ref{sm:data} of the Supplementary Information).

\begin{assumption}
\label{as:main}
Assume that $| \mathbf x_i \cdot \mathbf x_i| = 1$ for all $i$. For some $0 < \delta < 1$, let $|\mathbf x_i \cdot \mathbf x_j| \leq 1 - \delta$ for all $i, j \in [n]$ with $i \neq j$. Moreover, assume $\delta = \Omega(1/\mathrm{poly}\,n)$ for a dataset of size $n$.
\end{assumption}

Finally, for the neural network with activation function $\sigma$, we require a normalization constraint equivalent to applying batchnorm after every layer of the neural network.

\begin{assumption}
\label{as:norm}
The activation function $\sigma\,:\,\mathbb{R}\to\mathbb{R}$ is normalized such that
\begin{align}
\underset{X\sim\mathcal{N}(0, 1)}{\mathbb{E}} [\sigma(X)] = 0 \text{  and  } \underset{X\sim\mathcal{N}(0, 1)}{\mathbb{V}} [\sigma(X)] = \underset{X\sim\mathcal{N}(0, 1)}{\mathbb{E}} [\sigma^2(X)] = 1.
\end{align}
\end{assumption}

Following~\citet{agarwal2020deep}, we define the nonlinearity of the activation function and note the effect of normalization on the resulting constant.

\begin{definition}[Coefficient of nonlinearity]
\label{def:mu}
The coefficient of nonlinearity of the activation function $\sigma$ is defined to be $\mu := 1 - \left(\mathbb{E}_{X\sim\mathcal{N}(0,\, 1)}[X \sigma(X)]\right)^2$.
\end{definition}

To provide some intuition for the NTK, we briefly comment on the output distribution of a trained neural network in the limit of $t\to\infty$. Consider a test data example $\mathbf x_* \in \mathbb{R}^n$ and a training set $S$ of $n$ data examples. We denote the vector of evaluations of the kernel $k(\mathbf x_*, \mathbf x_i)$ between $\mathbf x_*$ and each $\mathbf x_i \in S$ by $\mathbf{k}_* \in \mathbb{R}^n$, and similarly for the covariance $(\mathbf{k}_\mathrm{cov})_*$. Define $n\times n$ matrices $K$ and $K_\mathrm{cov}$ as the evaluation of the NTK and covariance over the training set, where the $(i, j)$th element corresponds to an evaluation between examples $\mathbf x_i, \mathbf x_j \in S$. Since the neural network is initialized as a Gaussian distribution and the NTK describes an affine transform, a neural network with linearized dynamics (i.e. in the wide limit) will have Gaussian-distributed output. In particular, Corollary 1 of~\citet{NEURIPS2019_0d1a9651} gives the mean and variance of the Gaussian output $f(\mathbf x_*)$ of the converged NTK as $t\to\infty$:
\begin{align}
    \mathbb{E}[f(\mathbf x_*)] &= \mathbf{k}_*^T K^{-1}\mathbf y\\
    \mathbb{V}[f(\mathbf x_*)] &= K_\mathrm{cov}(\mathbf{x}_*, \mathbf{x}_*) + \mathbf{k}_*^T K^{-1}K_{\mathrm{cov}} K^{-1}\mathbf{k}_* - (\mathbf{k}_*^T K^{-1}(\mathbf{k}_\mathrm{cov})_* + h.c.),
\end{align}
where $h.c.$ denotes the Hermitian conjugate. Fully characterizing the output of a trained wide neural network thus consists of computing these two quantities; computing the mean only requires evaluating $\mathbf{k}_*^T K^{-1}\mathbf y$.

\subsection{Computing elements of the NTK}
\label{sm:ntk:comp}
To write the elements of the NTK, we define the dual activation function $\hat\sigma$ corresponding to the activation function $\sigma$~\cite{NIPS2016_abea47ba}.

\begin{definition}
\label{def:dual}
Consider data $\mathbf x_i, \mathbf x_j \in \mathbb{R}^d$ such that $||\mathbf x_i|| = ||\mathbf x_j|| = 1$ and hence $\rho = \mathbf x_i \cdot \mathbf x_i \in [-1, 1]$. Define the conjugate activation function $\hat\sigma: [-1, 1] \rightarrow [-1, 1]$ as follows:
\begin{align}
    \hat\sigma(\mathbf x_i \cdot \mathbf x_j) := \mathbb{E}_{\mathbf w \sim \mathcal{N}(0,\, I_d)}[\sigma(\mathbf w \cdot \mathbf x_i)\sigma(\mathbf w \cdot \mathbf x_j)].
\end{align}
\end{definition}

From~\citet{NEURIPS2019_dbc4d84b}, the elements of the NTK are given in terms of the $\ell$-fold composition $\hat{\sigma}^{(\ell)}$ of the dual activation function. Specifically, for $\rho_{ij} = \mathbf x_i \cdot \mathbf x_j$,
\begin{align}
\label{eq:ntk}
    K_{ij} &= \sum_{h=1}^{L+1} \hat\sigma^{(h-1)}(\rho_{ij})\left(\prod_{h'=h}^L \hat{\dot{\sigma}}(\hat\sigma^{(h')}(\rho_{ij}))\right).
\end{align}
Throughout the text, we will use $K_{ij}$ to denote the $(i, j)$th matrix element $k(\mathbf x_i, \mathbf x_j)$. Because Eq.~\ref{eq:ntk} only requires the inner product $\mathbf x_i \cdot \mathbf x_j$, we will also define for convenience the function $\hat k(\mathbf x_i \cdot \mathbf x_j) := k(\mathbf x_i, \mathbf x_j)$.

Data separability ensures that a single element of the NTK can be computed in time $O(\log n)$.

\begin{lemma}[Efficient NTK element computation]
\label{lem:eff}
If $L = \Theta(L_\mathrm{conv})$ and $\delta = \Omega(1/\mathrm{poly}\,n)$, then an element of the NTK can be computed in $O(\log (n)/\mu)$ time given the inner product between two data points.
\end{lemma}
\begin{proof}
By Eq.~\ref{eq:ntk}, a polynomial number of operations in $L$ are required to evaluate the NTK between two data points. Since $L_\mathrm{conv} = O(\log(n/\delta)/\mu)$, choosing $\delta=O(1/\mathrm{poly}\,n)$ ensures that $L=\Theta(L_\mathrm{conv}) = \Theta(\log(n)/\mu)$. Thus, an element of the NTK matrix can be computed in $O(\log(n)/\mu)$ time given the inner product between data. An example of a dataset satisfying this condition is described in the main text and further discussed in Section~\ref{sm:data}.
\end{proof}

Finally, we note some important properties of the conjugate activation function (Def.~\ref{def:dual}) from~\citet{NIPS2016_abea47ba} and~\citet{agarwal2020deep}.

\begin{remark}
\label{rem:dual}
The following properties hold for an activation function normalized under Assumption~\ref{as:norm}, where $h_0, h_1, \dots$ denote the Hermite polynomials.
\begin{enumerate}
    \item Let $a_i = \mathbb{E}_{z\sim\mathcal{N}(0,1)}[\sigma(z)h_i(z)]$. Due to normalization of $\sigma$, $a_0 = 0$ and $\sum_{i=1}^\infty a_i^2 = 1$.
    \item We have Hermite expansions $\sigma(u) = \sum_{i=1}^\infty a_i h_i(u)$ and $\hat\sigma(\rho) = \sum_{i=1}^\infty a_i^2 \rho^i$.
    \item Due to normalization of $\sigma$, we have $0 < \mu \leq 1$ and in particular $\mu = 1 - a_1^2$.
    \item If $\dot\sigma$ denotes the derivative of $\sigma$, then $\hat{\dot{\sigma}} = \dot{\hat{\sigma}}$.
\end{enumerate}
\end{remark}

\subsection{Neural network depth for the convergence of the NTK}
\label{sm:ntk:conv}
The results in the main text are given in terms of $L_\mathrm{conv}$, which arises naturally as the minimum depth for the neural network to converge by gradient descent. This minimum depth is given by
\begin{align}
    L_\mathrm{conv} &:= \frac{8\log(n/\delta)}{\mu},
\end{align}
which we later show to be the threshold for a well-condition NTK.

We now provide a more formal description of the convergence properties of a deep neural network. Following standard convention, we assume a squared loss function $\ell(\hat y, y) = (\hat y - y)^2$ defines the empirical loss function over the neural network parameterized by weights $\vec W$:
\begin{align}
    \mathcal{L}(\vec W) := \frac{1}{n} \sum_{i=1}^n \ell(f_{\vec W}(\mathbf x_i), \,y_i).
\end{align}
For smooth activation functions, we use a result of~\citet{NEURIPS2019_0d1a9651}, reproduced here directly from~\citet{agarwal2020deep}:

\begin{theorem}[Convergence via gradient descent]
\label{thm:lee}
Suppose that the activation $\sigma$ and its derivative $\sigma'$ further satisfy the properties that there exists a constant $c$, such that for all $\mathbf x_i, \mathbf x_j$ 
\[|\sigma(\mathbf x_i)|, |\sigma'(\mathbf x_i)|, \frac{|\sigma'(\mathbf x_i)-\sigma'(\mathbf x_j)|}{|\mathbf x_i - \mathbf x_j|} \leq c.\]
Then there exists a constant $N$ (depending on L, n, $\delta$) such that for width $m > N$ and setting the learning rate $\eta  = 2(\lambda_{\min}K + \lambda_{\max}K)^{-1}$, with high probability over the initialization the following is satisfied for gradient descent for all $t$, 
\[\mathcal{L}(\vec W(t)) \leq e^{- \Omega\left(\frac{t}{\kappa(K)}\right)} \mathcal{L}(\vec W(0))\]
\end{theorem}

It thus suffices to show that the NTK is well-conditioned for $L = \Omega\left(\frac{\log(n/\delta)}{\mu}\right)$ in order to show that $\mathcal{L}(\vec W(t))$ to converge to $\mathcal{L}(\vec W(0))$ via gradient descent. In the remainder of this section, we demonstrate the well-conditioning of the NTK for $L \geq L_\mathrm{conv}$, obtaining a condition number bounded by $1 \leq \kappa \leq \frac{1+1/n}{1-1/n}$ in Corollary~\ref{cor:cond}. Hence, we conclude that $L = \Omega\left(\frac{\log(n/\delta)}{\mu}\right) = \Omega(L_\mathrm{conv})$ is a sufficient neural network depth for the neural network to provably converge via gradient descent. Throughout this paper, we will apply this particular depth to demonstrate the correspondence of the regime of efficient gradient descent with the regime in which a quantum speedup is obtained.

\subsection{Bounding elements of the NTK}
\label{sm:ntk:bounds}
We will require bounds on the matrix elements of the NTK in the regime $L \geq L_\mathrm{conv}$ of convergence by gradient descent. For convenience, we define a function $B(L, \delta, \mu)$: for any $\mu \in (0,1]$, $\delta \in (0, 1)$ and a positive integer $L$, we let
\begin{align}
    B(L, \delta, \mu) &:= \frac{1}{2}\left(1 - \frac{\mu}{2}\right)^{L-L_0(\delta, \mu)}, \text{ where}\\
    L_0(\delta, \mu) &:= \max\left\{\left\lceil\frac{\log(\frac{1}{2\delta})}{\log(1+\frac{\mu}{2})}\right\rceil, \, 0\right\} = O\left(\frac{\log(1/\delta)}{\mu}\right).
\end{align}
These are related to the convergence depth $L_\mathrm{conv}$ by the following lemma.

\begin{lemma}
\label{lem:l0}
For all $\mu, \delta$ and $n \geq 2$, we have that $L_\mathrm{conv} \geq 2L_0(\delta, \mu)$.
\end{lemma}
\begin{proof}
Since $L_\mathrm{conv} > 0$, this is trivially satisfied if $\delta \geq 1/2$. For $\delta < 1/2$, we note that the derivative of $L_\mathrm{conv}/L_0$ with respect to $\delta$ is given by
\begin{align}
    \frac{\partial}{\partial \delta}\frac{L_\mathrm{conv}}{L_0} &= \frac{4\log(2n)\log(1+\mu/2)}{\delta \mu \log^2(2\delta)} > 0,
\end{align}
and thus evaluating the limiting case of $\delta \to 0$ is sufficient to bound $L_\mathrm{conv}/L_0$. Since $\lim_{\delta\to0} L_\mathrm{conv}/L_0 = \frac{4 \log(1+\mu/2)}{\mu} > 1$, we conclude that $L_\mathrm{conv} \geq 2L_0(\delta, \mu)$ for all allowed parameter values.
\end{proof}

The proof of an upper bound may be found in Theorem 27 of~\citet{agarwal2020deep}, the result of which we state here.

\begin{theorem}[NTK element upper bound]
\label{thm:kbound}
Consider an NTK corresponding to a neural network of depth $L$. The diagonal entries of $K$ are all equal and given by $K_{ii} = \frac{\hat{\dot{\sigma}}(1)^{L+1}-1}{\hat{\dot{\sigma}}(1) - 1}$. Furthermore, if $L \geq 2L_0(\delta, \mu)$, then
\begin{align}
\left|\frac{K_{ij}}{K_{11}}\right| &\leq 2B(L/2, \delta_{ij}, \mu),
\end{align}
where $|\mathbf x_i \cdot \mathbf x_j| = 1 - \delta_{ij}$ for $i\neq j$.
\end{theorem}

As a result of Lemma~\ref{lem:l0}, we observe that Theorem~\ref{thm:kbound} on the matrix elements of the NTK is valid for all $L \geq L_\mathrm{conv}$. Moreover, we can simplify the bound further.

\begin{lemma}
\label{lem:bbound}
If $L \geq L_\mathrm{conv}$, then we have the following bounds on $K_{ij}$ for $i \neq j$. If $0 < \delta_{ij} < 1/2$
\begin{align}
    \left|\frac{K_{ij}}{K_{11}}\right| &\leq \left(\frac{\delta}{\delta_{ij} n}\right)^2,
\end{align}
while for $1/2 \leq \delta_{ij} \leq 1$,
\begin{align}
    \left|\frac{K_{ij}}{K_{11}}\right| &\leq \left(\frac{\delta}{n}\right)^2.
\end{align}
\end{lemma}
\begin{proof}
Since $\mu \in (0, 1]$, we find that for $\delta_{ij} < 1/2$
\begin{align}
    L_0(\delta_{ij}, \mu) = \max\left\{\left\lceil\frac{\log(\frac{1}{2\delta_{ij}})}{\log(1+\frac{\mu}{2})}\right\rceil, \, 0\right\} \leq \frac{5\log(1/\delta_{ij})}{2\mu},
\end{align}
while for $\delta_{ij} \geq 1/2$ we have $L_0(\delta_{ij}, \mu) = 0$. Accordingly, we can weaken the bound on $2B(L/2, \delta, \mu)$ when the depth is set to $L = \alpha L_\mathrm{conv} \geq \frac{8 \log(n/\delta)}{\mu}$, i.e. $\alpha \geq 1$. Taking the more nontrivial case of $\delta_{ij} < 1/2$, we have
\begin{align}
    2B(L/2, \delta, \mu) &= \left(1 - \frac{\mu}{2}\right)^{\frac{8\alpha \log(n/\delta)}{2\mu} - \frac{5\log(1/\delta_{ij})}{2\mu}} \leq \left(\frac{\delta}{n}\right)^{-8\alpha\log(1-\mu/2)/2\mu} \left(\frac{1}{\delta_{ij}}\right)^{-5\log(1-\mu/2)/2\mu}.
\end{align}
Since $\mu \in (0, 1]$, we can take limiting cases of the exponents and observe that $\alpha=1$ places the loosest bound. This gives for $\delta_{ij} < 1/2$,
\begin{align}
    2B(L/2, \delta, \mu) \leq \left(\frac{\delta}{n}\right)^2 \left(\frac{1}{\delta_{ij}}\right)^2.
\end{align}
Repeating the analysis with $L_0 = 0$ for $\delta_{ij} \geq 1/2$, we have $2B(L/2, \delta, \mu) \leq \left(\frac{\delta}{n}\right)^2$. Applying Theorem~\ref{thm:kbound} and noting that $L_0(\delta, \mu) \geq L_0(\delta_{ij}, \mu)$ when applying Lemma~\ref{lem:l0}, these results correspond to bounds on the NTK matrix element.
\end{proof}

Since $\delta = \min_{i,j}\delta_{ij}$, we have the following corollary.

\begin{corollary}[Deep NTK upper bound]
\label{cor:kbound}
If $L \geq L_\mathrm{conv}$, then $\left|\frac{K_{ij}}{K_{11}}\right| \leq \frac{1}{n^2}$ for all $i \neq j$.
\end{corollary}

\subsection{Eigenvalue bounds of the NTK}
\label{sm:sub:eigs}
We require bounds on the maximum and minimum eigenvalues of the NTK in order to compute error bounds on the NTK approximation evaluated by the quantum algorithm. These directly follow from the matrix element bounds of the NTK.

\begin{lemma}[Maximum eigenvalue of NTK]
\label{lem:eigmax}
If $L \geq L_\mathrm{conv}$, then $\lambda_\mathrm{max}K \leq K_{11}(1 + 1/n)$.
\end{lemma}
\begin{proof}
As given by Theorem~\ref{thm:kbound}, the diagonal elements of the NTK are equal and larger than the off-diagonal elements, since Lemma~\ref{lem:l0} guarantees sufficient neural network depth $L_\mathrm{conv} \geq 2L_0(\delta, \mu)$. By the Gershgorin circle theorem, $\lambda_\mathrm{max} \leq K_{11}[1 + (n-1)(2B(L/2, \delta,\mu))]]$. Applying Corollary~\ref{cor:kbound}, this gives an upper bound of $\lambda_\mathrm{max} \leq K_{11}[1 + (n-1)/n^2] \leq K_{11}(1 + 1/n)$.
\end{proof}

\begin{lemma}[Minimum eigenvalue of NTK]
\label{lem:eigmin}
If $L \geq L_\mathrm{conv}$, then $\lambda_\mathrm{min}K \geq K_{11}(1 - 1/n)$.
\end{lemma}
\begin{proof}
Similarly to above, the Gershgorin circle theorem with Corollary~\ref{cor:kbound} gives $\lambda_\mathrm{min}K \geq K_{11}[1 - (n-1)/n^2] \geq K_{11}(1 - 1/n)$.
\end{proof}

From these bounds, we conclude that the NTK is well-conditioned when representing a neural network deep enough to converge, consistent with the result of~\citet{agarwal2020deep}.

\begin{corollary}[Conditioning of NTK]
\label{cor:cond}
The condition number $1 \leq \kappa(K) \leq \frac{1 + 1/n}{1 - 1/n}$ converges to unity as $n \to \infty$.
\end{corollary}

\section{Approximations of the NTK}
\label{sm:approx}
To bound the error caused by sparsifying the NTK or replacing it with a diagonal matrix, we require a result on matrix inverses (see~\citet{doi:10.1137/0613003} for a proof).
\begin{lemma}[Perturbation of matrix inverses]
\label{lem:inv}
Let $A$ be an $n \times n$ real matrix. A small perturbation $\epsilon X$ to $A$ causes a small perturbation of $A^{-1}$ bounded in spectral norm by
\begin{align}
    \frac{||(A + \epsilon X)^{-1} - A^{-1}||}{||A^{-1}||} \leq \kappa(A) \cdot \frac{||\epsilon X||}{||A||} + O(||\epsilon X||^2).
\end{align}
\end{lemma}

\subsection{Diagonal NTK approximation}

\begin{theorem}[Convergence of the diagonal NTK to the exact NTK]
\label{thm:identity}
Let $M = K_{11} \cdot I$ be proportional to the $n \times n$ identity matrix. The error of the matrix inverse vanishes as $\frac{||M - K^{-1}||}{||K^{-1}||} = O(1/n)$.
\end{theorem}
\begin{proof}
Define $n \times n$ matrix $A = K / K_{11}$ and let $\epsilon X = I - A$. Since $A$ has a unit diagonal, $\epsilon X$ has a zero diagonal. By Corollary~\ref{cor:kbound}, all elements of $\epsilon X$ are bounded in magnitude by $1/n^2$. By the Gershgorin circle theorem, the maximum eigenvalue of $X$ is thus $1/n$. Applying the results of Section~\ref{sm:sub:eigs} and Lemma~\ref{lem:inv}, we find that
\begin{align}
\frac{||(A + \epsilon X)^{-1} - A^{-1}||}{||A^{-1}||} \leq \frac{1 + 1/n}{1 - 1/n} \cdot \frac{1/n}{1 + 1/n} + O(1/n^2) = O(1/n).
\end{align}
Since $K = K_{11} A$, this gives the required relation for the NTK itself. Hence, the error vanishes rapidly with a polynomial increase in dataset size.
\end{proof}

\begin{corollary}
\label{cor:identity}
For a training dataset of size $n$, the expectation of an infinite-width neural network $f$ of depth $L \geq L_\mathrm{conv}$ on test data $\mathbf x_*$ can be estimated as $\mathbb{E}[f_*] \approx \mathbf{k}_*^T \mathbf{y}$ up to $O(1/n)$ error.
\end{corollary}

\subsection{Sparsified NTK approximation}
\begin{lemma}[Maximum eigenvalue of sparsified NTK]
\label{lem:sp-eigmax}
If $L \geq L_\mathrm{conv}$, then $\lambda_\mathrm{max}(\tilde K) \leq K_{11}(1 + 1/n)$.
\end{lemma}
\begin{proof}
As given by Theorem~\ref{thm:kbound}, the diagonal elements of the NTK are equal and larger than the off-diagonal elements, since Lemma~\ref{lem:l0} guarantees sufficient neural network depth $L_\mathrm{conv} \geq 2L_0(\delta, \mu)$. By the Gershgorin circle theorem, $\lambda_\mathrm{max} \leq K_{11}[1 + s(2B(L/2, \delta,\mu))]]$. Applying Corollary~\ref{cor:kbound}, this gives an upper bound of $\lambda_\mathrm{max} \leq K_{11}[1 + s/n^2] \leq K_{11}(1 + 1/n)$ since $s = O(\log n)$.
\end{proof}

\begin{lemma}[Minimum eigenvalue of sparsified NTK]
\label{lem:sp-eigmin}
If $L \geq L_\mathrm{conv}$, then $\lambda_\mathrm{min}(\tilde K) \geq K_{11}(1 - 1/n)$.
\end{lemma}
\begin{proof}
Similarly to above, the Gershgorin circle theorem with Corollary~\ref{cor:kbound} gives $\lambda_\mathrm{min}K \geq K_{11}[1 - s/n^2] \geq K_{11}(1 - 1/n)$.
\end{proof}

From these bounds, we conclude that the NTK is well-conditioned when representing a neural network deep enough to converge, similarly to Corollary~\ref{cor:cond}.

\begin{corollary}[Conditioning of sparsified NTK]
The condition number $1 \leq \kappa(\tilde K) \leq \frac{1 + 1/n}{1 - 1/n}$ converges to unity as $n \to \infty$.
\end{corollary}

Hence, the sparsified NTK has a condition number that is well-suited to running HHL. Finally, we show that it converges to the exact NTK.

\begin{theorem}[Convergence of the sparsified NTK to the exact NTK]
\label{thm:sparse}
Let $M = \tilde K$ be a sparsification of the exact NTK $K$ with the complete diagonal and any subset of $s = O(\log n)$ off-diagonal elements. The error of the matrix inverse vanishes as $\frac{||\tilde K^{-1} - K^{-1}||}{||K^{-1}||} = O(1/n)$.
\end{theorem}
\begin{proof}
Define matrices $A = K / K_{11}$ and $\tilde A = \tilde K / K_{11}$. Let $\epsilon X = \tilde A - A$. Since $A$ and $\tilde A$ both have unit diagonal, $\epsilon X$ has a zero diagonal. By Corollary~\ref{cor:kbound}, all elements of $\epsilon X$ are bounded in magnitude by $1/n^2$. By the Gershgorin circle theorem, the maximum eigenvalue of $X$ is thus $1/n$. Applying the results of Section~\ref{sm:sub:eigs} and Lemma~\ref{lem:inv}, we find that
\begin{align}
\frac{||(A + \epsilon X)^{-1} - A^{-1}||}{||A^{-1}||} \leq \frac{1 + 1/n}{1 - 1/n} \cdot \frac{(n-\log n)/n^2}{1 + 1/n} + O(1/n^2) = O(1/n).
\end{align}
Since $K = K_{11} A$, this gives the required relation for the NTK itself. Hence, the error vanishes rapidly with a polynomial increase in dataset size.
\end{proof}

Since we sparsify the NTK instead of replacing by a diagonal matrix, the above results show that the error bound of the sparsified NTK provided by the Gershgorin circle theorem is slightly tighter than the diagonal approximation despite the same asymptotic error of $O(1/n)$, by a factor of $O(\log (n)/n^2)$. Although we do not place a lower bound on the improvement of the sparsified NTK compared to the diagonal NTK, the numerical experiments provided in Figure~\ref{fig:my:perf} of the main text suggest that the sparsified NTK does provide a significantly better approximation of the exact neural network output.

\section{Quantum algorithm}
\label{sm:qntk}

To evaluate the necessary matrix inversions and inner products described in Algorithm~\ref{alg} of the main text, we require several standard quantum linear algebra routines.

\subsection{Quantum random access memory}
A key feature of attractive applications in quantum machine learning is achieving polylogarithmic dependence on training set size. However, the initial encoding of a training set trivially requires linear time, since each data example must be recorded once. To ensure that this linear overhead only occurs a single time, quantum random access memory (QRAM) can be used to prepare a classical data structure once and then efficiently read out data with quantum circuits in logarithmic time. We use the binary tree QRAM subroutine proposed by~\citet{qram} and applied commonly in quantum machine learning~\cite{NEURIPS2019_16026d60,Kerenidis2020Quantum}. The QRAM consists of a classical data structure that encodes a data matrix $S \in \mathbb{R}^{n \times d}$ with efficient \emph{quantum access}.

\begin{definition}[Quantum access]
\label{def:qacc}
Let $\ket{S_i} = \frac{1}{||S_i||} \sum_{j=0}^{d-1} S_{ij}\ket{j}$ denote the amplitude encoding of the $i$th row of data $S \in \mathbb{R}^{n \times d}$. Quantum access provides the mappings
\begin{itemize}
    \item $\ket{i}\ket{0} \mapsto \ket{i}\ket{S_i}$
    \item $\ket{0} \mapsto \frac{1}{||S||_F} \sum_i ||S_i|| \ket{i}$
\end{itemize}
in time $T$ for $i \in [n]$.
\end{definition}

The QRAM by~\citet{qram} provides quantum access in time $T$ that is polylogarithmic complexity with respect to both $n$ and $d$.

\begin{theorem}[QRAM]
\label{thm:qram}
For $S \in \mathbb{R}^{n \times d}$, there exists a data structure that stores $S$ such that the time to insert, update or delete entry $S_{ij}$ is $O(\log^2(n))$. Moreover, a quantum algorithm with access to the data structure provides quantum access in time $O(\mathrm{polylog}(nd))$.
\end{theorem}

Because the mapping $\ket{i}\ket{0} \mapsto \ket{i}\ket{S_i}$ is efficient, we can prepare a uniform superposition $\sum_{i=0}^{n-1} \ket{i}\ket{0} \mapsto \sum_{i=0}^{n-1}\ket{i}\ket{S_i}$ of the entire dataset. While preparing an arbitrary superposition is difficult, a uniform superposition is achieved with a constant-depth quantum circuit by applying Hadamard gates to all qubits. Hence, after a single $O(n)$ operation to prepare the data structure in QRAM, the dataset can be efficiently accessed by a quantum computer.

In our application of QRAM, we need to prepare states $\ket{x} = \frac{1}{\sqrt n}\sum_{i=0}^{n-1} \ket{x_i}$ and $\ket{y} = \frac{1}{\sqrt{n}} \sum_{i=0}^{n-1} y_i \ket{i}$. For the state $\ket{x}$, Assumption~\ref{as:norm} ensures that $||\mathbf x_i|| = 1$, allowing $\ket{x}$ to be directly prepared. For labels $y_i$, the classification problem ensures a known normalization factor $\sqrt{n}$.
 
\subsection{Preparation of kernel states}
To evaluate the neural network's prediction for a test data point $\mathbf x_*$, the NTK must be evaluated between $\mathbf x_*$ and the entire training set $\{\mathbf x_i\}$. In particular, we must prepare the quantum state $\ket{k_*} = \sum_{i=0}^{n-1} k_i \ket{i}$, where $k_i$ corresponds to an encoding proportional to kernel element $k(\mathbf x_*, \mathbf x_i)$ up to error $\epsilon$.

Since the NTK is only a function of the inner product $\rho_i = \mathbf x_\star \cdot \mathbf x_i$, we can use previous work on inner product estimation~\cite{NEURIPS2019_16026d60} to construct the kernel elements. By preparing this inner product in a quantum register, the NTK --- which is efficient to compute classically on a single pair of data points by since $\delta = \Omega(1/\mathrm{poly}\;n)$ --- can be efficiently evaluated between the test data point and the entire training dataset. However, we first need the well-known subroutines of amplitude estimation~\cite{brassard2002quantum} and median evaluation~\cite{wiebe2014quantum} as well as a basic translation from bitstring representations to amplitudes.

\begin{lemma}[Amplitude estimation]
\label{lem:amp}
Consider a quantum algorithm $A : \ket{0} \mapsto \sqrt{p}\ket{v, 1} + \sqrt{1 - p}\ket{g, 0}$ for some garbage state $\ket{g}$. For any positive integer $P$, amplitude estimation outputs $\tilde p \in [0, 1]$ such that
\begin{align}
    |\tilde p - p | \leq 2\pi \frac{\sqrt{p(1-p)}{P}} + \left(\frac{\pi}{P}\right)^2
\end{align}
with probability at least $8/\pi^2$ using $P$ iterations of the algorithm $A$. If $p=0$, then $\tilde p = 0$ with certainty, and similarly for $p = 1$.
\end{lemma}

\begin{lemma}[Median evaluation]
\label{lem:median}
Consider a unitary $U : \ket{0^{\otimes m}} \mapsto \sqrt{\alpha}\ket{v, 1} + \sqrt{1 - \alpha}\ket{g, 0}$ for some $1/2<\alpha \leq 1$ in time $T$. Then there exists a quantum algorithm that, for any $\Delta > 0$ and for any $1/2 < \alpha_0 \leq \alpha$, produces a state $\ket{\psi}$ such that $|| \ket{\psi} - \ket{0^{\otimes mL}}\ket{x} || \leq \sqrt{2\Delta}$ for some integer $L$ in time
\begin{align}
    2T\left\lceil \frac{\log(1/\Delta)}{2(|\alpha_0| - 1/2)^2} \right\rceil.
\end{align}
\end{lemma}

\begin{lemma}[Amplitude encoding]
\label{lem:enc}
Given state $\frac{1}{\sqrt{n}} \sum_{i=0}^{n-1} \ket{k_i}$ with $0 \leq k_i \leq 1$, the state $\frac{1}{\sqrt{P}} \sum_{i=0}^{n-1} k_i \ket{i}$ may be prepared in time $O(1/P)$ with $P = \sum_{i=0}^{n-1}k_i^2$.
\end{lemma}
\begin{proof}
We consider a single element $\ket{k_i}$ in the superposition $\frac{1}{\sqrt{n}} \sum_{i=0}^{n-1} \ket{k_i}$. Adding an ancilla to perform the map $\ket{k_i}\ket{0} \mapsto \ket{k_i}\ket{\arccos k_i}$, each bit of the binary expansion $\ket{k_i}\ket{\arccos k_i} = \ket{k_i}\ket{b_1}\dots \ket{b_m}$ can be used as a rotation angle. Specifically, insert the ancilla $\frac{1}{\sqrt{2}}(\ket{0} + \ket{1})$ and apply $m$ controlled rotations $\exp(i b_j \sigma^z/2^j)$ to obtain the state $\ket{k_i} \ket{\arccos k_i}(|k_i| \ket{0} + \sqrt{1-k_i^2}\ket{1})$. By including an additional rotation controlled on the sign of $k_i$, the state $k_i \ket{0} + \sqrt{1-k_i^2}\ket{1}$ can be prepared. Applying the above in superposition, we have the state $\frac{1}{\sqrt{n}} \sum_{i=0}^{n-1} \ket{i} \left(k_i \ket{0} + \sqrt{1 - k_i^2}\ket{1}\right)$. Letting $P = \sum_{i=0}^{n-1}k_i^2$, post-selection on the final state gives $\frac{1}{\sqrt{P}} \sum_{i=0}^{n-1} k_i \ket{i}$ in time $O(1/P)$.
\end{proof}

Lemmas~\ref{lem:amp} through~\ref{lem:enc} provide the basic quantum computing routines required to prepare the necessary states. To achieve appropriate normalization, we need to assume a reasonable upper bound on an NTK element between the test data point and any data point in the training set. Since the NTK is entirely determined by the inner product between data points, i.e. $\hat k(\mathbf x_* \cdot \mathbf x_i) := k(\mathbf x_*, \mathbf x_i)$, it suffices to create an estimate $\hat\delta$ of the training set separability $\delta$ then set the normalization threshold to $\hat k(1 - \hat\delta)$. As shown in Section~\ref{sm:data}, $\delta(n)$ typically follows a power law, allowing $\hat\delta$ to be chosen without sampling large datasets.

\begin{theorem}[Kernel estimation]
\label{thm:kprep}
Let $S \in \mathbb{R}^{n \times d}$ be the training dataset of $\{\mathbf x_i\}$ unit norm vectors stored in the QRAM described in Theorem~\ref{thm:qram}. Consider the neural tangent kernel described in Eq.~\ref{eq:ntk} with coefficient of nonlinearity $\mu$. For test data vector $\mathbf x_* \in \mathbb{R}^d$ in QRAM and constant $\hat\delta$, there exists a quantum algorithm that maps
\begin{align}
    \frac{1}{\sqrt{n}} \sum_{i=0}^{n-1} \ket{i}\ket{0} \mapsto \frac{1}{\sqrt{n}} \sum_{i=0}^{n-1} \ket{i}\ket{k(\mathbf x_*, \mathbf x_i)}
\end{align}
in time $O(\mathrm{polylog}(nd))$ for $k(\mathbf x_*, \mathbf x_i)$ expressed up to some finite precision.

To prepare an amplitude-encoded state we require normalization by a maximum kernel element $k_\mathrm{max} \approx \max_i k(\mathbf x_*, \mathbf x_i)$. Given an estimator $\hat\delta$ of the dataset separability, $k_\mathrm{max} \approx \hat k(1 - \hat\delta)$. For $k_i := \hat k(\rho_i) / \hat k(1 - \hat\delta)$ clipped to $-1 \leq k_i \leq 1$, the state
\begin{align}
    \frac{1}{\sqrt{n}} \sum_{i=0}^{n-1} \ket{i}\ket{0} \mapsto \frac{1}{\sqrt{P}} \sum_{i=0}^{n-1} k_i \ket{i}
\end{align}
may be prepared with error $|\rho_i - \mathbf x_* \cdot \mathbf x_i | \leq \epsilon$ with probability $1 - 2\Delta$ in time $\tilde O(\mathrm{polylog}(nd)\log(1/\Delta)/\epsilon)$.
\end{theorem}
\begin{proof}
Since the NTK is only a function of the inner product $\mathbf x_* \cdot \mathbf x_i$, we can compute the kernel elements after estimating the inner product between the test data and training data, following a similar approach to~\citet{NEURIPS2019_16026d60}. Consider the initial state $\ket{i}\frac{1}{\sqrt{2}}(\ket{0} + \ket{1})\ket{0}$. Using the QRAM as an oracle controlled on the second register, we can in $O(\mathrm{polylog}(nd))$ time map $\ket{i}\ket{0}\ket{0} \mapsto \ket{i}\ket{0}\ket{x_i}$ and similarly $\ket{i}\ket{1}\ket{0} \mapsto \ket{i}\ket{1}\ket{x_*}$. (If $\ket{x_*}$ is not in QRAM, this operation only takes $O(d)$ time.) Applying a Hadamard gate on the second register, the state becomes
\begin{align}
    \frac{1}{2}\ket{i}\left(\ket{0}(\ket{x_i}+\ket{x_*}) + \ket{1}(\ket{x_i}-\ket{x_*})\right).
\end{align}
Measuring the second qubit in the computational basis, the probability of obtaining the $\ket{1}$ state is $p_i = \frac{1}{2}(1 - \bra{x_i}\ket{x_\star})$ since the vectors are real-valued. Writing the state $\ket{1}(\ket{x_i} - \ket{x_*})$ as $\ket{v_i, 1}$, we have the mapping
\begin{align}
    A : \ket{i}\ket{0} \mapsto \ket{i}(\sqrt{p_i}\ket{v_i, 1} + \sqrt{1 - p_i}\ket{g_i, 0}),
\end{align}
where $\ket{g_i}$ is a garbage state. The runtime of $A$ is $\tilde O(\mathrm{polylog}(nd))$. This completes the first statement to prepare the kernel elements of the matrix.

We now turn to preparing the amplitude-encoded state $\ket{k_*}$. Applying amplitude estimation with $A$, we obtain a unitary $U$ that performs
\begin{align}
    U : \ket{i}\ket{0} \mapsto \ket{i}\left(\sqrt{\alpha}\ket{\tilde p_i, g, 1} + \sqrt{1-\alpha}\ket{g', 0}\right)
\end{align}
for garbage registers $g, g'$. By Lemma~\ref{lem:amp}, we have $|\tilde p_i - p_i| \leq \epsilon$ and $8/\pi^2 \leq \alpha \leq 1$ after $O(1/\epsilon)$ iterations. At this point, we now have runtime $\tilde O(\mathrm{polylog}(nd)/\epsilon)$.

Applying median estimation, we finally obtain a state $\ket{\psi_i}$ such that $|| \ket{\psi_i} - \ket{0}^{\otimes L} \ket{\tilde p_i, g}|| \leq \sqrt{2\Delta}$ in runtime $\tilde O(\mathrm{polylog}(nd)\log(1/\Delta)/\epsilon)$. Performing this entire procedure but on the initial superposition $\sum_{i=0}^{n-1} \ket{i}\frac{1}{\sqrt{2}}(\ket{0} + \ket{1})\ket{0}$, we now have the final state $\sum_{i=0}^{n-1} \ket{\psi_i}$.

Since $\left|\tilde p_i - \frac{1 - \mathbf x_* \cdot \mathbf x_i}{2} \right| \leq \epsilon$, we can recover the inner product $\mathbf x_* \cdot \mathbf x_i$ as a quantum state. In general, there exists a unitary $V : \sum_x \ket{x, 0} \mapsto \sum_x \ket{x, f(x)}$ for any classical function $f$ with the same time complexity as $f$. Hence, we can choose $f$ that recovers $\mathbf x_* \cdot \mathbf x_i \approx 1 - 2\tilde p_i$ up to $O(\epsilon)$ error with probability $1 - 2\Delta$. Because $\delta = \Omega(1/\mathrm{poly}\;n)$, evaluating the NTK between two data points takes time $O(\mathrm{polylog}(n)/\mu)$ given their inner product. Again evaluating the classical function, we obtain the state $\frac{1}{\sqrt{n}}\sum_{i=0}^{n-1}\ket{i}\ket{k_i}$ where $k_i$ has $\leq O(\epsilon)$ error in time $\tilde O(\mathrm{polylog}(nd)\log(1/\Delta)/\epsilon\mu)$.

Finally, we need to prepare the state $\ket{k_*} = \frac{1}{\sqrt{P}} \sum_{i=0}^{n-1} k_i \ket{i}$ for $P = \sum_i k_i^2$. Applying Lemma~\ref{lem:enc}, preparing $\ket{k_*}$ requires $O(1/P)$ time.
\end{proof}

\subsection{Quantum linear systems algorithm}
\label{sm:qlsa}
To invert the sparsified NTK matrix, we must solve the quantum linear systems problem (QLSP).

\begin{definition}[QLSP]
Let $A$ be an $n\times n$ Hermitian matrix with condition number $\kappa$, unit determinant, and at most $s$ nonzero entries in any row or column. Let $\mathbf x, \mathbf b$ be $n$-dimensional vectors such that $\mathbf x = A^{-1} \mathbf b$. We define the quantum states $\ket{b}, \ket{x}$ such that
\begin{align}
    \ket{b} := \frac{\sum_{i=1}^n b_i \ket{i}}{||\sum_{i=1}^n b_i \ket{i}||} \quad \text{and} \quad \ket{x} := \frac{\sum_{i=1}^n x_i \ket{i}}{||\sum_{i=1}^n x_i \ket{i}||}.
\end{align}
Given access to a procedure $\mathcal{P}_A$ that computes the locations and values of the nonzero entries in $A$, and a procedure $\mathcal{P}_B$ that prepares the state $\ket{b}$ in $O(\mathrm{polylog}(n))$ time, output a state $\ket{\tilde x}$ such that $||\ket{\tilde x}-\ket{x}|| \leq \epsilon$, succeeding with probability larger than $1/2$ and providing a flag indicating success.
\end{definition}

To \emph{classically} solve QLSP, a computational cost of at least $O(n)$ is required for a sparse, well-conditioned and positive definite $n\times n$ linear system. In particular, the conjugate gradient method~\cite{10.5555/865018} achieves $O(n s \sqrt{\kappa} \log(1/\epsilon))$ time to precision $\epsilon$ for a positive definite matrix. Proposed by Harrow, Hassidim and Lloyd, the HHL algorithm~\cite{PhysRevLett.103.150502} obtains an exponential speedup over this result.

\begin{theorem}[HHL algorithm]
\label{thm:hhl}
The quantum linear systems problem for $s$-sparse matrix $A \in \mathbb{R}^{n\times n}$ can be solved by a gate-efficient algorithm (i.e. with only logarithmic overhead in gate complexity) that makes $O(\kappa^2 s \mathrm{poly}(\log(s \kappa/\epsilon)/\epsilon))$ queries to the oracle $\mathcal{P}_A$ of the matrix $A$ and $O(\kappa s \mathrm{poly}(\log(s\kappa/\epsilon))/\epsilon)$ queries to the oracle to prepare the state corresponding to $\mathbf{b}$. Using a quantum random access memory for data access contributes a multiplicative factor of $O(\log n)$ to the runtime.
\end{theorem}

Note that by replacing the phase estimation subroutine with algorithms based on Chebyshev polynomial decompositions, the dependence on precision can be improved. Similarly, improvements on the Hamiltonian simulation subroutine further improve the dependence on sparsity~\cite{7354428}. Based on these extension to HHL, QLSP can be solved in $O(\log(n) \kappa s \,\mathrm{polylog}(\kappa s/\epsilon))$ time~\cite{doi:10.1137/16M1087072}.

Such quantum linear systems algorithms can solve the problem of sparse matrix inversion, which is known to be BQP-complete.

We now turn to the issue of constructing a sparse matrix $\tilde{K}$ with a logarithmic number of nonzero elements in any row or column. To apply HHL, we need an efficient oracle $\mathcal{P}_A$ as required by Theorem~\ref{thm:hhl}, which must report nonzero indices of any column in logarithmic time. As described by~\citet{doi:10.1137/16M1087072}, the procedure $\mathcal{P}_A$ must perform the map $\ket{j, \ell} \mapsto \ket{j, \nu(j, \ell)}$ for any $j \in [n]$ and $\ell \in [s]$. The function $\nu:[n]\times[s]\to [n]$ computes the row index of the $\ell$th nonzero entry of the $j$th column. (For simplicity, we assume the sparsity pattern is symmetric; since it is chosen deterministically, this may easily be enforced.) Additionally, the procedure must perform the map $\ket{j, k, z} \mapsto \ket{j, k, z \oplus A_{jk}}$ for $j, k \in [N]$, where $A_{jk}$ is a bitstring representation of the $jk$th element of the matrix $A$.

Using the first statement of Theorem~\ref{thm:kprep} and a fixed sparsity pattern, the efficiency of computing an NTK element (Lemma~\ref{lem:eff}) ensures that the procedure $\mathcal{P}_A$ requires only logarithmic time in the training set size. We discuss the choice of sparsity pattern in Section~\ref{sm:num}, which is chosen to enforce $s = O(\log n)$. Hence, given sufficient well-conditioning $\kappa = O(\log n)$ of the sparsified NTK, solving the QLSP takes polylogarithmic time in the training set size.

\subsection{Readout}
To approximate the output of the NTK, we must evaluate the sign of either $\mathbf{k}_*^T \mathbf{y}$ (corresponding to the diagonal NTK) or $\mathbf{k}_*^T \tilde K \mathbf{y}$ (corresponding to the sparsified NTK). Thus far, we have described the preparation of the quantum states $\ket{k_*}$ and $\tilde K^{-1} \ket{k_*}$.
To estimate the sign of the inner product $\bra{y}\ket{k_*}$ or $\bra{y}\tilde K^{-1}\ket{k_*}$, we encode the relative phase the states and perform an inner product estimation procedure such that $m$ measurements of the state gives $1/\sqrt{m}$ error~\cite{zhao2019compiling}.

\begin{lemma}[Inner product estimation]
\label{lem:inner}
Given states $\ket{s}, \ket{y} \in \mathbb{R}^n$, estimating $\bra{s}\ket{y}$ with $m$ measurements has standard deviation at most $1/\sqrt{m}$. Here, we take either $\ket{s} = \ket{k_*}$ or $\ket{s} = \tilde K^{-1} \ket{k_*}$ (i.e. the diagonal or sparsified NTK approximations).
\end{lemma}
\begin{proof}
Prepare initial state $\frac{1}{\sqrt{2}}(\ket{0}\ket{s}+\ket{1}\ket{y})$. Applying a Hadamard gate to the first qubit, we obtain state $\frac{1}{2}(\ket{0}(\ket{s} + \ket{y}) + \ket{1}(\ket{s}-\ket{y}))$. Measuring the first qubit, the probability of obtaining $\ket{0}$ is $p = \frac{1}{2}(1 + \bra{s}\ket{y})$. The binomial distribution over $m$ trials given this probability has variance $mp(1-p)$, and thus the variance of the estimate of $p$ is $p(1-p)/m$. Transforming to get the variance of the overlap estimate $\bra{s}\ket{y}$, we find a standard deviation of $\sqrt{\frac{1 - (\bra{s}\ket{y})^2}{m}}$. Since $(\bra{s}\ket{y})^2 \leq 1$, this is upper-bounded by $1/\sqrt{m}$.
\end{proof}

\subsection{Extensions to other neural network architectures}
\label{sm:qntk:cnn}
To apply Algorithm~\ref{alg} of the main text to a general neural network architecture, a few key properties are required. First, the sparsified NTK matrix must be well-conditioned at some neural network depth $L(n)$; second, the computation of the NTK between two data points must be efficient at that depth; finally, the data distribution must permit efficient post-selection and readout. Due to the structure of the NTK matrix, it may be that the sparsified and/or diagonal NTK approximations converge to the exact neural network output as $n$ increases. The generality of this result may be expected to extend to a broad set of neural network architectures, since the well-conditioned NTK corresponds to efficient trainability by gradient descent. Satisfying the first condition is thus likely due to the success of deep neural networks; satisfying the third condition is shown to be true here for the MNIST dataset and common choices of neural network architectures.

In the case of the fully-connected neural network, we showed that an NTK matrix element of a network with depth $L = O(\log n)$ can be efficiently computed if $\delta = \Omega(1/\mathrm{poly}\;n)$, which is commonly satisfied for real-world datasets (see Section~\ref{sm:data} below). In the case of a convolutional neural network, the provided numerical experiments use a fixed depth of $L = 101$ and establish conditioning by exaggerating the vanishing of off-diagonal elements (see Section~\ref{sm:num:cnn} for further discussion). For fixed $L$, a matrix element of the convolutional NTK is trivially efficient to compute: similarly to Eq.~\ref{eq:ntk} for the fully-connected network, the NTK of a convolutional neural network only depends on the data dimension $d$ and neural network depth $L$. Hence, the time to compute a matrix element is independent of the training set size $n$. Given the numerical evidence of the conditioning of the convolutional NTK with increasing depth (Figure~\ref{fig:my:cond:cnn}), it may be relevant to consider a broader set of neural network architectures whose depth to converge efficiently by gradient descent scales favorably as $L = O(\log n)$, since they would remain efficiently computable by a similar quantum algorithm (see for example~\citet{NEURIPS2019_dbc4d84b} for a convolutional NTK with time complexity $O(Ld)$ per matrix element).

% As shown numerically in Section~\ref{sm:num:cnn}, the convolutional NTK is conditioned with increasing $L$. If $L = O(\log n)$ is sufficiently deep to guarantee convergence by gradient descent (i.e. a well-conditioned NTK), then a kernel with convolutional operations remains efficient to compute. For a convolutional filter $\mathbf{w} \in \mathbb{R}^{q\times q}$ and an image $\mathbf x \in \mathbb{R}^{P\times Q}$, the convolution operator is defined as~\cite{NEURIPS2019_dbc4d84b}
% \begin{align}
%     [\mathbf w * \mathbf x]_{i,j} = \sum_{a = -\frac{q-1}{2}}^\frac{q-1}{2} \sum_{b=-\frac{q-1}{2}}^\frac{q-1}{2} [\mathbf w]_{a + \frac{q+1}{2}, b+\frac{q+1}{2}}[\mathbf x]_{a+i, b+j}, \;\mathrm{for}\;i\in[P],j\in[Q],
% \end{align}
% where $[\cdot]_{i,j}$ denotes the $(i, j)$th element. The kernel 

\section{Datasets}
\label{sm:data}

To efficiently compute the NTK between data $\mathbf x_i, \mathbf x_j$ as is necessary to achieve an exponential speedup over gradient descent, we require $\delta = \Omega(1/\mathrm{poly}\,n)$. Using a uniform distribution on a sphere, we motivate the power law $\delta(n) \approx a_1 n^{-a_2}$ with positive constants. We show such power laws to empirically hold on the MNIST handwritten digit image dataset, demonstrating that the requirement of $\delta = \Omega(1/\mathrm{poly}\; n)$ is satisfied for common datasets.

Define a dataset of $n$ training examples $(\mathbf{x}_i, y_i)$, where $\mathbf{x}_i \in \mathbb{R}^d$ has fixed dimension and $y_i$ is bounded. Each $\mathbf{x}_i$ is sampled uniformly on the sphere $S^{d-1}$. We can define $\delta$ in terms of an $n\times n$ matrix $G$ defined similarly to the Gram matrix but with magnitudes of inner products, i.e. $G_{ij} = |\mathbf x_i \cdot \mathbf x_j|$. The minimum dataset separability is given by $1 - \rho_\mathrm{max}$, where $\rho_\mathrm{max}$ is the largest off-diagonal element of $G$.

Since the elements of $G$ are not drawn independently from a single distribution, we instead define a symmetric $n\times n$ matrix $A$ with elements drawn i.i.d. from the distribution of inner product magnitudes. We show that a power law $\delta(n) = a_1 n^{-a_2}$ is satisfied for the matrix $A$.

\begin{lemma}
\label{lem:dsphere}
Let $A$ be a symmetric $n \times n$ matrix with elements $A_{ij}$ sampled from the distribution of inner products $|\mathbf x_i \cdot \mathbf x_j|$. Each matrix element is sampled i.i.d. with $\mathbf x_i, \mathbf x_j$ drawn uniformly at random from $S^{d-1}$ with $d \geq 3$. In the limit of large $n$, the separability $\delta = \min_{i,j} (1 - A_{ij})$ is lower-bounded by $\delta = \Omega(1/\mathrm{poly}\,n)$. In particular, $\delta \approx \Omega(n^{4/(1-d)})$ to leading order in large $n$.
\end{lemma}

\begin{proof}
We first determine the CDF of $|\mathbf x_i \cdot \mathbf x_j|$ for $\mathbf x_i, \mathbf x_j \in \mathbb{R}^d$ drawn uniformly at random from $S^{d-1}$. Without loss of generality, let $\mathbf x_i = (1,0, \dots, 0)$. Since the distribution is uniform on the surface of a sphere, symmetry under orthogonal matrix multiplication implies that we can let $\mathbf{x}_j = \frac{\mathbf{u}}{||\mathbf{u}||}$ for $\mathbf{u} \sim \mathcal{N}_d(0, 1)$ for $\mathbf u = (u_1, \dots, u_d)$. Hence, the distribution of $\mathbf x_i \cdot \mathbf x_j$ is equivalent to the distribution of $\rho \sim u_1/\sqrt{u_1^2 + \dots + u_d^2}$. Considering the random variable $\rho^2$, rearranging terms gives a ratio of $\chi^2$ variables and hence an $F$-distribution with 1 and $d-1$ degrees of freedom. Evaluating the CDF for $\rho = |\mathbf x_i \cdot \mathbf x_j|$ in terms of the hypergeometric $_2F_1$ function gives
\begin{align}
    F(|\rho|) &= \frac{2\Gamma\left(\frac{d}{2}\right)}{\sqrt{\pi}\Gamma\left(\frac{d-1}{2}\right)} {_2F_1}\left(\frac{1}{2}, \frac{3-d}{2}; \frac{3}{2}; |\rho|^2\right) |\rho|.
\end{align}

Suppose we sample from the distribution $m$ times, corresponding to the $m \approx n^2/2$ randomly chosen elements in the symmetric matrix $A$. To find the largest $|\rho|$ corresponding to the minimum separability, we seek the $(1 - 1/m)$th percentile of the $m$ elements. Following Mosteller's work on order statistics~\cite{10.1214/aoms/1177730881}, the largest $|\rho|$ will be asymptotically normally distributed for large $m$, with a mean of $F^{-1}(1 - 1/m)$. Since we expect $|\rho|$ to converge to 1, we Taylor expand $|\rho| = 1 - \delta$ around $\delta = 0$ to give
\begin{align}
    F(1 - \delta) \approx 1 + \frac{2^{\frac{d-1}{2}} \delta^{\frac{d-1}{2}} \Gamma \left(\frac{d}{2}\right)}{\sqrt{\pi } \Gamma \left(\frac{d+1}{2}\right)}.
\end{align}
Solving for $\delta$ in $F(1 - \delta) = 1 - 1/m$, we find that in expectation
\begin{align}
    \delta &= \pi ^{\frac{1}{d-1}} \left(\frac{2^{\frac{1}{2}-\frac{d}{2}} \Gamma \left(\frac{d+1}{2}\right)}{m \Gamma \left(\frac{d}{2}\right)}\right)^{\frac{2}{d-1}}.
\end{align}
Substituting back $m \approx n^2 / 2$ gives $\delta(n) = a_1 n^{-a_2}$ with $a_2 = 4/(d-1)$. Taking a bounding case on $d$, we conclude that $\delta = \Omega(n^{-2}) = \Omega(1/\mathrm{poly}\,n)$ for all $d \geq 3$.
\end{proof}

\begin{figure}[H]
  \centering
  \includegraphics[width=0.48\textwidth]{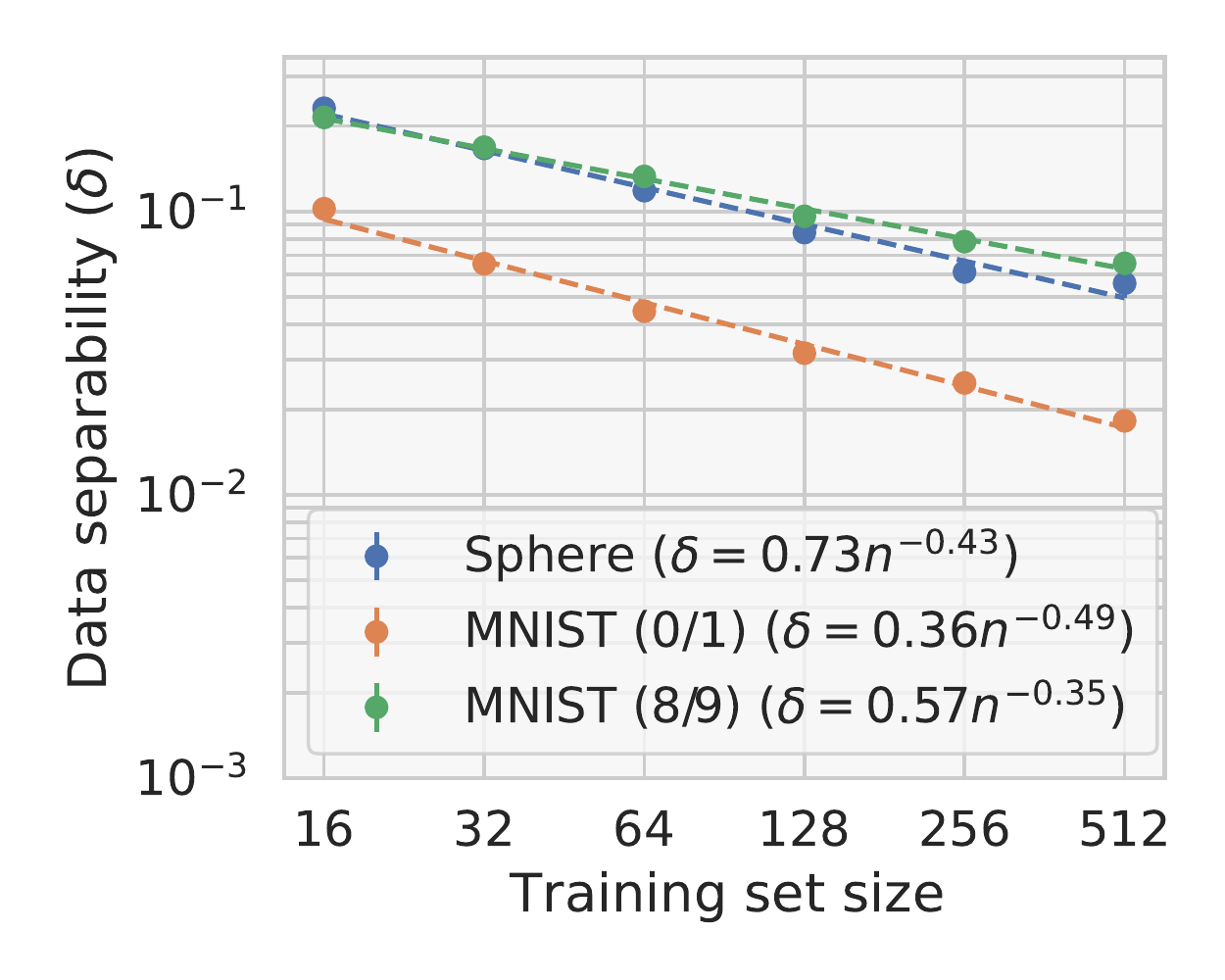}
  \caption{Empirical fit of $\delta(n)$ for the uniform sphere dataset with $d=10$ and for subsets of the MNIST dataset (0 vs. 1 binary classification and 8 vs. 9 binary classification). We find $\delta(n) \approx 0.73 n^{-0.43}$ ($R^2 = 0.98$) for the uniform distribution on the sphere, showing good agreement with the prediction of $\delta(n) \propto n^{-0.44}$ by Lemma~\ref{lem:dsphere}. The MNIST datasets similarly both fit with $R^2 = 0.99$.}
  \label{fig:delta}
\end{figure}

Although Lemma~\ref{lem:dsphere} addresses an independently sampled matrix of inner products, we confirm that it empirically describes the spherical dataset with good numerical agreement to the independent sampling approximation. Moreover, the same power law ansatz fits the MNIST dataset (Figure~\ref{fig:delta}), suggesting generality of the result to real-world datasets.

Due to the straightforward extrapolation of the scaling of $\delta$, the choice of normalization factor $k_\mathrm{max}$ necessary to prepare state $\ket{k_*}$ becomes efficient to estimate (see Theorem~\ref{thm:kprep}). In particular, the largest expected NTK between the test data point and any data point in the training set may be approximated by $k_\mathrm{max} = \hat k(1 - \hat \delta)$, where $\hat\delta$ is estimated by a power law to be sufficiently small given the training set size $n$. This allows the state $\ket{k_*}$ to be prepared without requiring any direct classical evaluation of the full training set.

\section{Performance of the quantum NTK}
\label{sm:num}

We discuss numerical experiments on the MNIST handwritten digit image classification dataset for the fully-connected neural network and the convolutional neural network with pooling layers. The output of either network is compared to the sparsified and diagonal NTK approximations, where the sparsified NTK requires a sparsity pattern such that the number of nonzero elements $s$ in any row or column is at most $s = O(\log n)$.

We briefly elaborate on the choice of sparsity pattern used for the numerical experiments. As described in Section~\ref{sm:qlsa}, there must be a function $\nu:[n]\times[s]\to [n]$ that computes the row index of the $\ell$th nonzero entry of the $j$th column in $O(\log n)$ time. When the dataset is initially stored in QRAM (incurring $\tilde O(n)$ cost), the function $\nu$ is also generated. We choose $s$ nonzero indices in $[n]$, corresponding to each row of the NTK over the $n \times n$ matrix. If the row index exceeds the column index, then discard the generated index; the resulting indices may then reflected to create a symmetric sparsity pattern, with a sparsity $s$ that in expectation remains $O(\log n)$. Since the NTK measures the similarity between data examples, the sparsity pattern may also be biased towards the most relevant matrix elements by setting all elements between classes to zero (as seen in Figure~\ref{fig:my:ntk} in the main text, which is illustrated over the sorted training set). The pattern is stored in a sorted data structure to allow $\nu$ to operate in logarithmic time by a binary search on the column index. Although this process incurs a one-time $O(n \log n)$ cost, the sparsity pattern then remains fixed for any choice of neural network (much like storing the dataset in QRAM). Multiple neural networks of different architectures may be trained using the same sparsity pattern, since the pattern is independent of the choice of neural architecture.

\subsection{Fully-connected neural network}

As discussed in the main text, a shallower network of $L = L_\mathrm{conv}/10$ layers is used. The depth $L_\mathrm{conv}$ to converge by gradient descent is determined by the rapidity at which the off-diagonal elements vanish (Figure~\ref{fig:ff:el}), causing the condition number to converge to unity (Corollary~\ref{cor:cond}).

\begin{figure}[H]
  \centering
  \begin{subfigure}[t]{0.33\textwidth}
  \includegraphics[width=\textwidth]{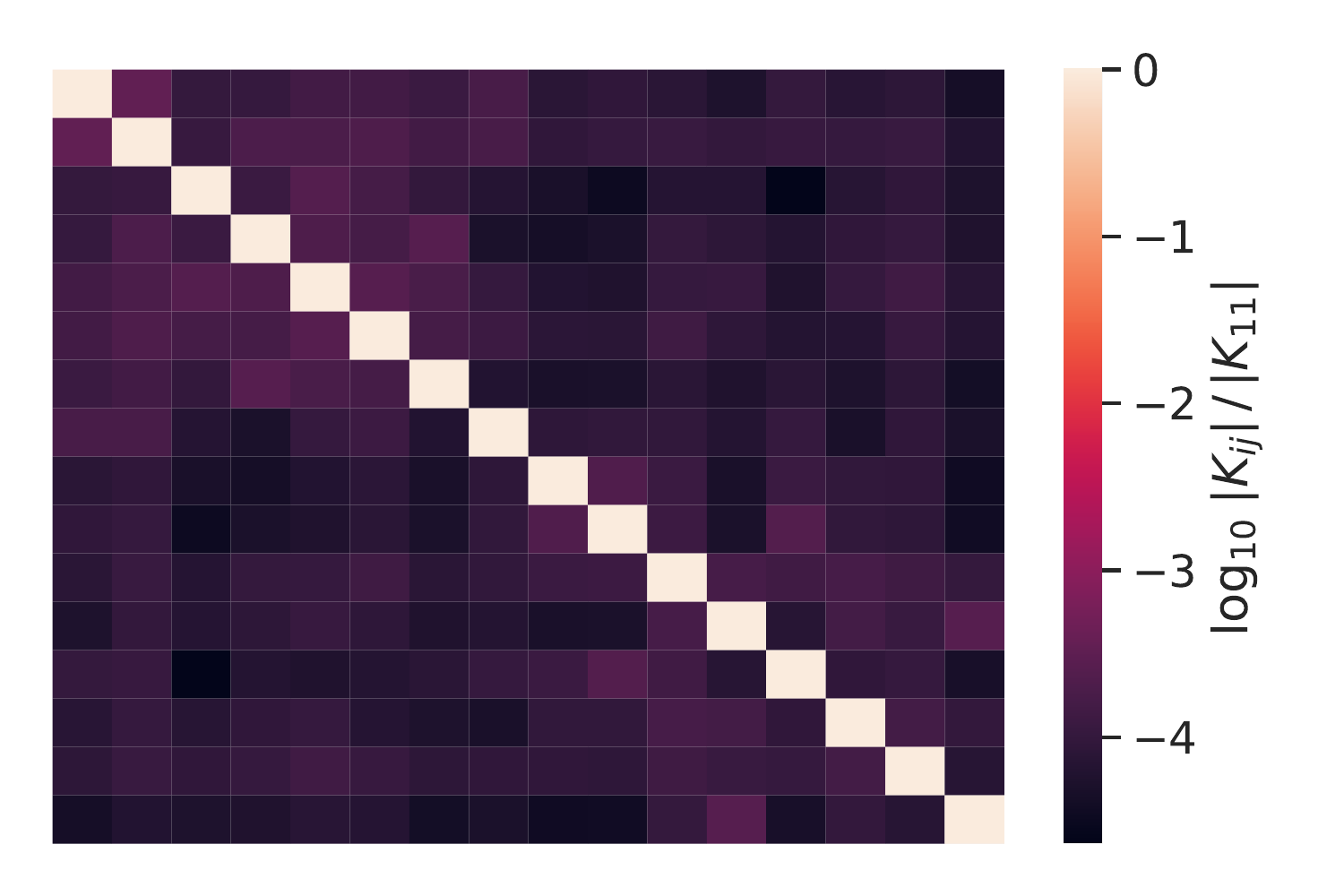}
  \caption{$n = 16$}
  \end{subfigure}%
  \begin{subfigure}[t]{0.33\textwidth}
  \includegraphics[width=\textwidth]{figs/sphere_ntk_N32.pdf}
  \caption{$n = 64$}
  \end{subfigure}%
  \begin{subfigure}[t]{0.33\textwidth}
  \includegraphics[width=\textwidth]{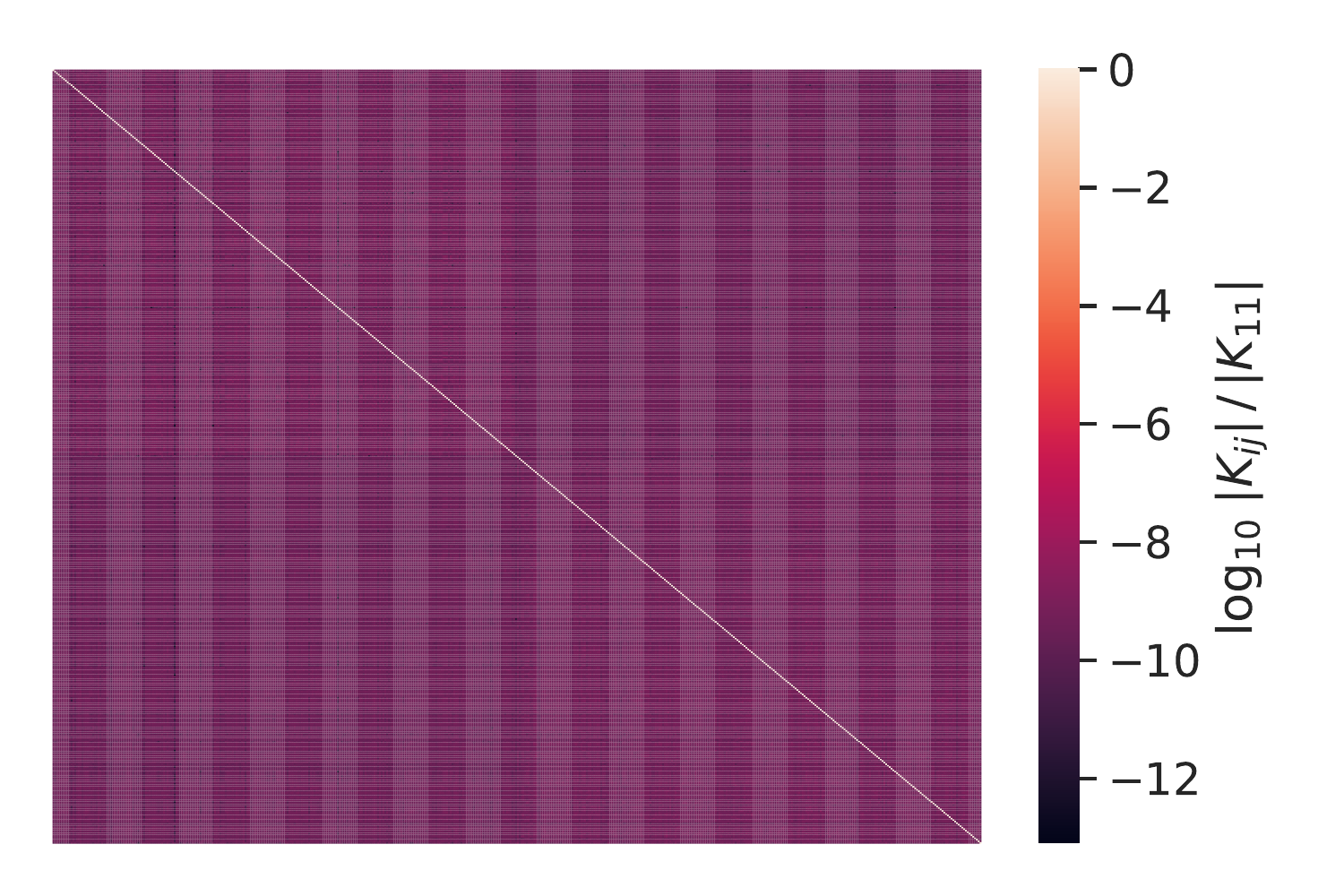}
  \caption{$n = 512$}
  \end{subfigure}
  \caption{\textit{Matrix elements of the fully-connected NTK at depth $L = L_\mathrm{conv}/10$.} The depth of the neural network scales logarithmically with the training set size since $\delta = 1/\mathrm{poly}\;n$ (Section~\ref{sm:data}). As it increases, the off-diagonal elements rapidly vanish (note the changing color scale).}
  \label{fig:ff:el}
\end{figure}

Empirically, the vanishing of off-diagonal elements occurs faster than the proven upper bound, causing the condition number to quickly approach unity at the shallower depth of $L_\mathrm{conv}/10$ (Figure~\ref{fig:ff:cond}). As discussed above, the sparsity pattern is selected such that there are $O(\log n)$ off-diagonal elements (Figure~\ref{fig:ff:sparse}).

\begin{figure}[H]
  \centering
  \begin{subfigure}[t]{0.45\textwidth}
  \includegraphics[width=\textwidth]{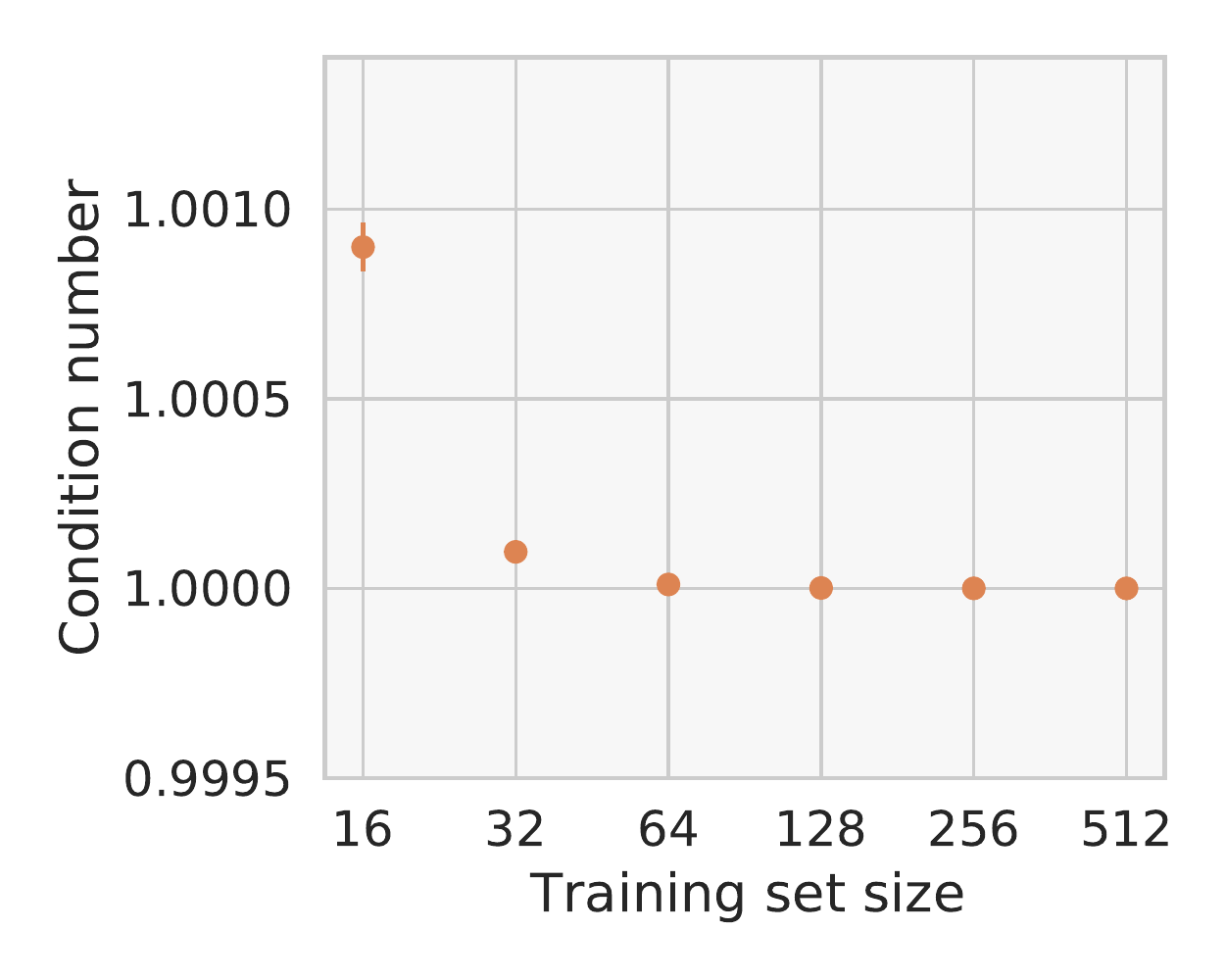}
  \caption{Condition number of the sparsified NTK\label{fig:ff:cond}}
  \end{subfigure}%
  \hspace{0.05\textwidth}%
  \begin{subfigure}[t]{0.45\textwidth}
  \includegraphics[width=\textwidth]{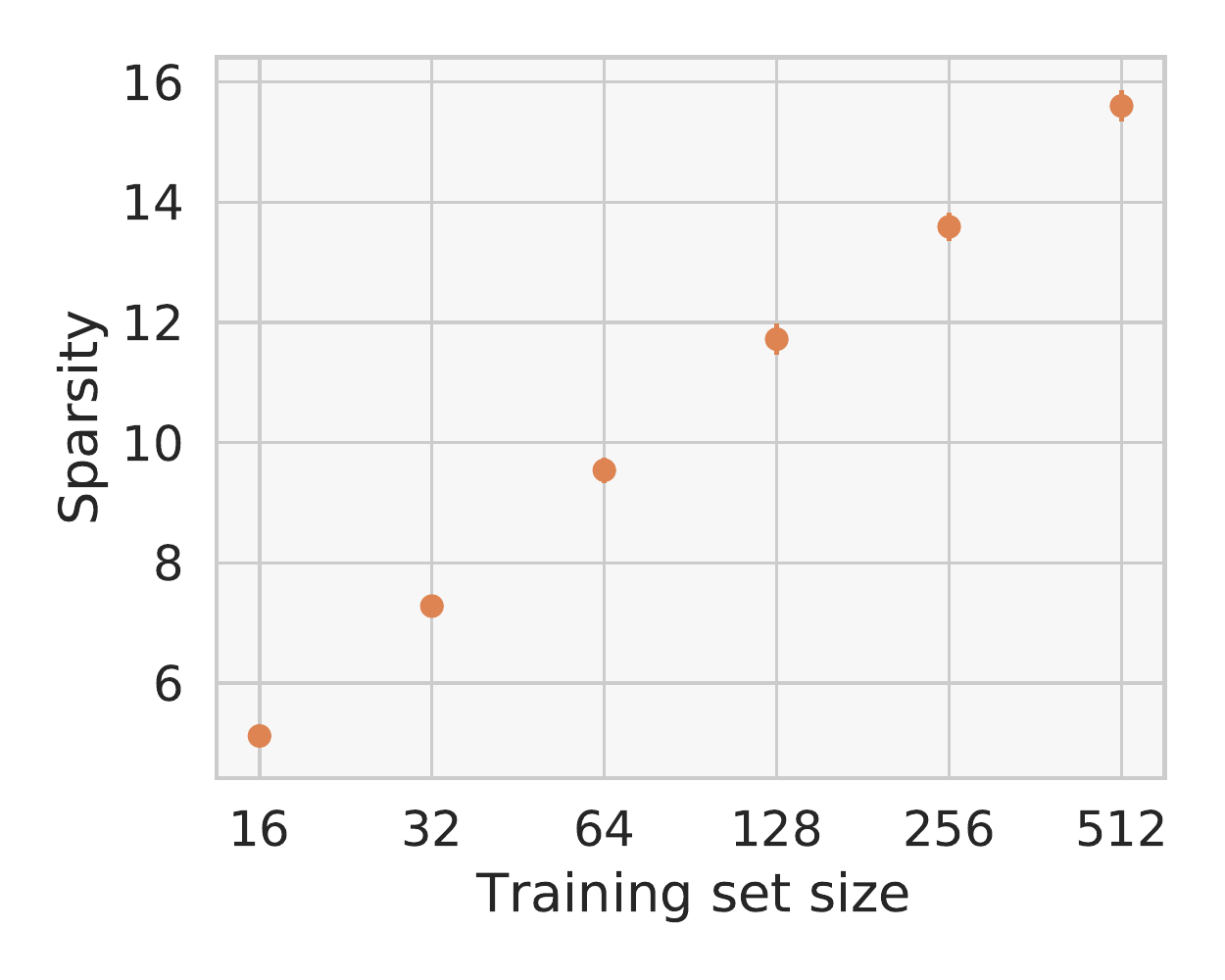}
  \caption{Sparsity of the sparsified NTK\label{fig:ff:sparse}}
  \end{subfigure}
  \caption{\textit{Efficiency of QLSP for the fully-connected sparsified NTK ($L = L_\mathrm{conv}/10$, 8 vs. 9 MNIST classification).} \textbf{(a)} As the training set size increases --- and consequently the neural network depth of $O(\log n)$ --- the condition number rapidly approaches 1, satisfying the bound $\kappa = O(\log n)$. \textbf{(b)} Maximum number of nonzero elements of the sparsified NTK matrix, verifying that $s = O(\log n)$ as required by the sparsity pattern.}
\end{figure}

\subsection{Convolutional neural network}
\label{sm:num:cnn}

The Myrtle convolutional neural network~\cite{myrtle} is modified for binary classification. The basic unit of the architecture consists of $L_c$ layers of $3\times 3$ convolutional filters with ReLU activation functions followed by a $2 \times 2$ pooling layer with a stride of 2. To classify the MNIST dataset ($28 \times 28$ pixel images), this unit is repeated three times and followed by a final pooling layer and fully connected layer. We use an erf activation function for the output layer to obtain output between $-1$ and $1$, allowing the class to be directly determined by measuring the sign of the output.

While we only proved the well-conditioning of the fully-connected network, the convolutional neural network with pooling is seen to share the same vanishing off-diagonal structure (Figure~\ref{fig:my:el}). Given the widespread success of \emph{deep} neural networks across different network architectures, it is perhaps unsurprising that the trainability of deep models by gradient descent extends to the Myrtle NTK. The NTK matrix structure of neural networks amenable to gradient descent motivates the application of the quantum algorithm to the convolutional neural network.

\begin{figure}[H]
  \centering
  \begin{subfigure}[t]{0.33\textwidth}
  \includegraphics[width=\textwidth]{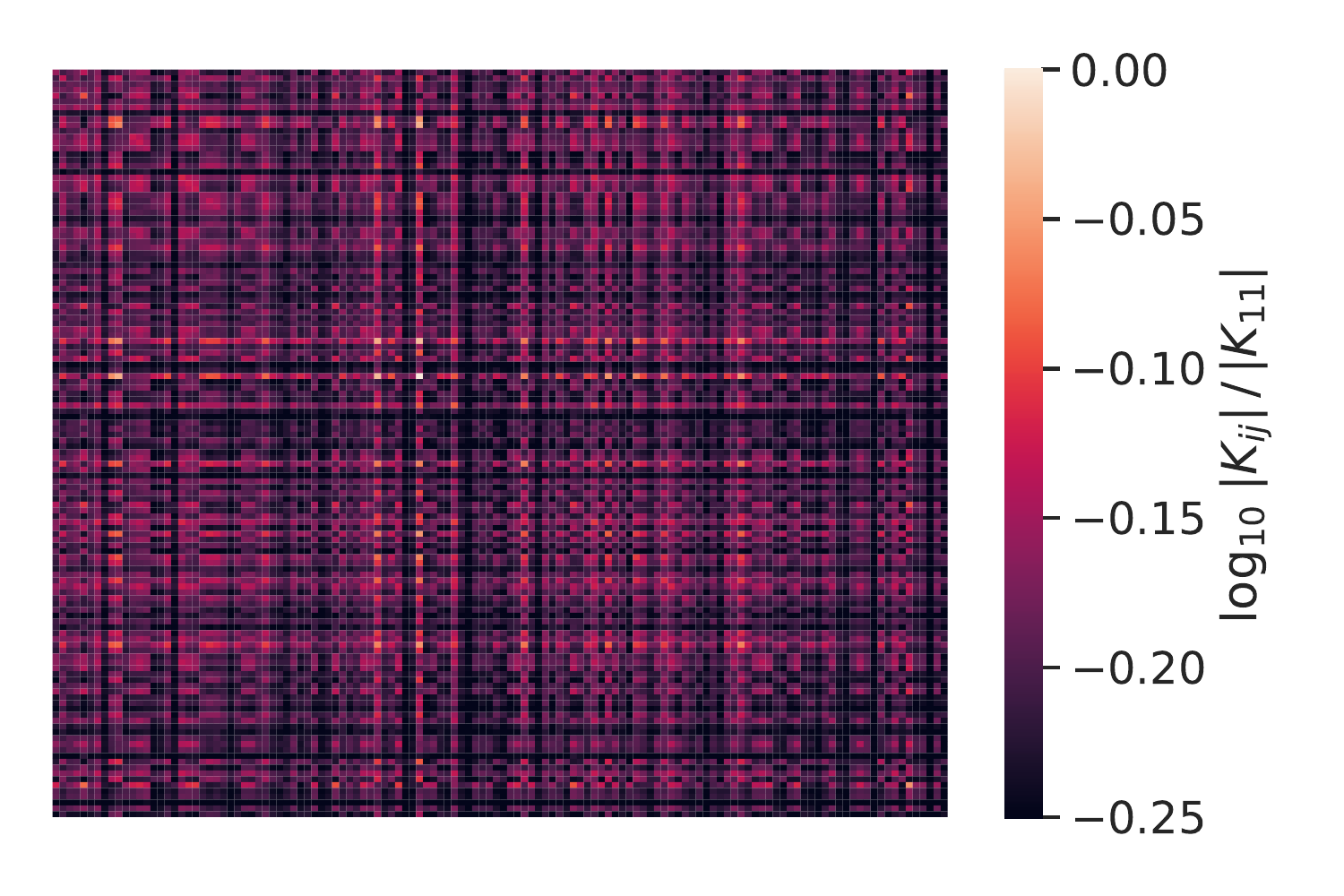}
  \caption{Depth 9}
  \end{subfigure}%
  \begin{subfigure}[t]{0.33\textwidth}
  \includegraphics[width=\textwidth]{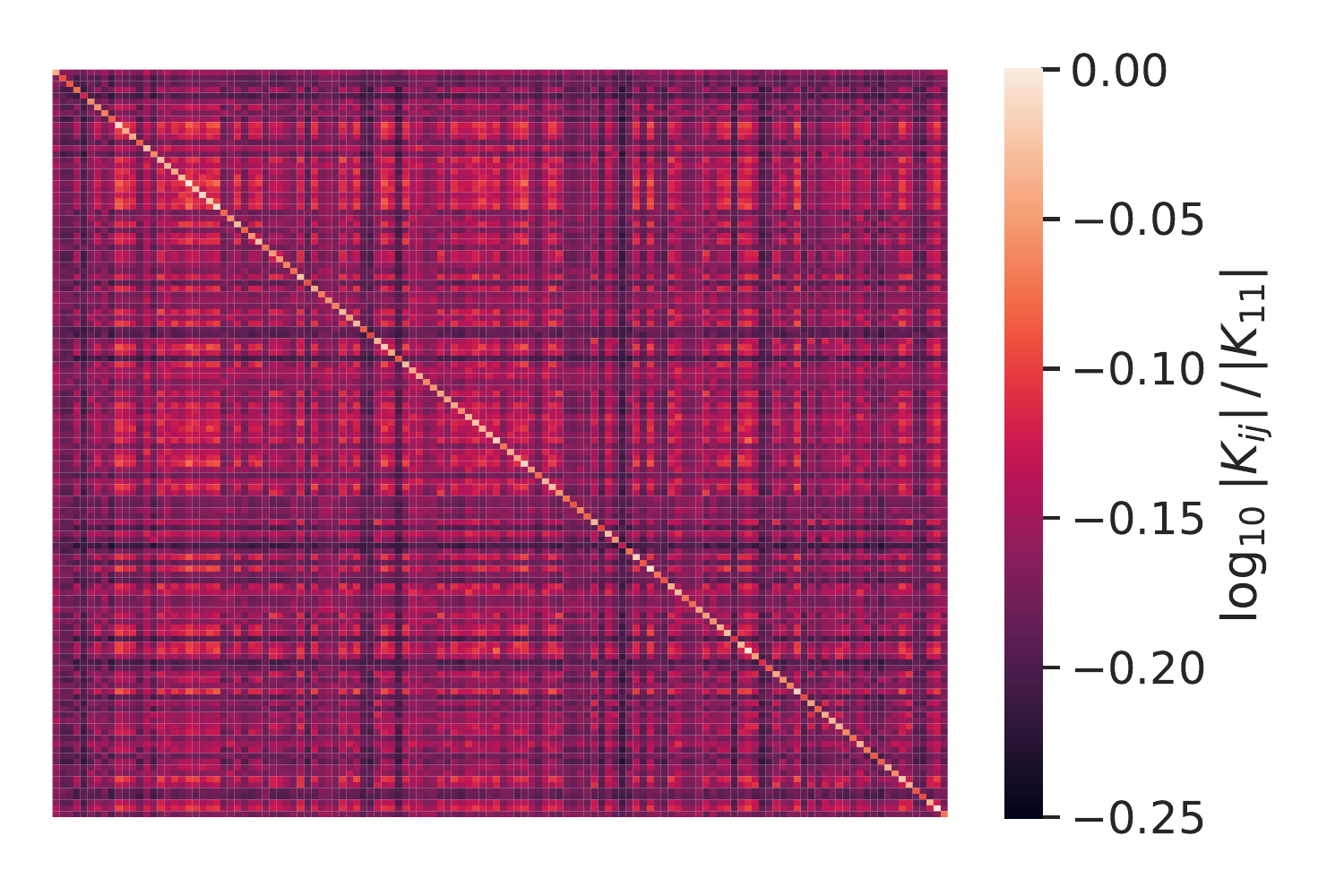}
  \caption{Depth 101}
  \end{subfigure}%
  \begin{subfigure}[t]{0.33\textwidth}
  \includegraphics[width=\textwidth]{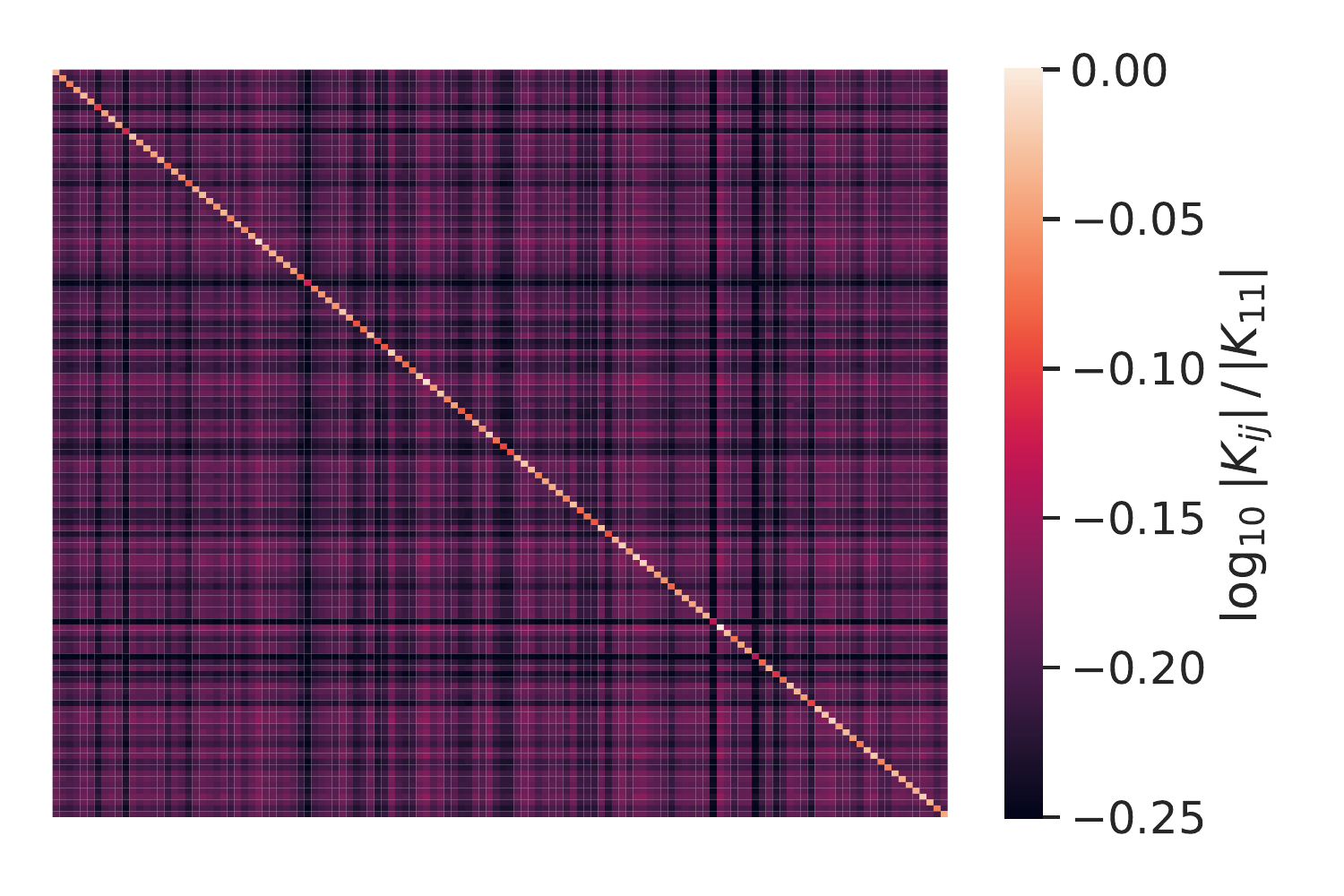}
  \caption{Depth 1001}
  \end{subfigure}
  \caption{\textit{Matrix elements of the convolutional NTK.} Applied to the MNIST dataset (0 vs. 1 binary classification, $n=64$), the convolutional neural network with pooling layers has vanishing off-diagonal elements as the depth increases, similarly to the proven fully-connected network.}
  \label{fig:my:el}
\end{figure}

Due to the vanishing off-diagonal elements, the conditioning of the convolutional neural network improves with depth, similarly to the fully-connected network (Figures~\ref{fig:my:cond:fc} and~\ref{fig:my:cond:cnn}). While the diagonal elements of the fully-connected NTK are equal, the convolutions and pooling in the convolutional NTK allows $k(\mathbf x_i, \mathbf x_i)$ to take different values even if $\mathbf x_i$ is normalized to be on the unit sphere. However, the condition number still approaches a constant as the depth increases, showing the well-conditioning feature of deep convolutional neural networks.

\begin{figure}[H]
  \centering
  \begin{subfigure}[t]{0.33\textwidth}
  \includegraphics[width=\textwidth]{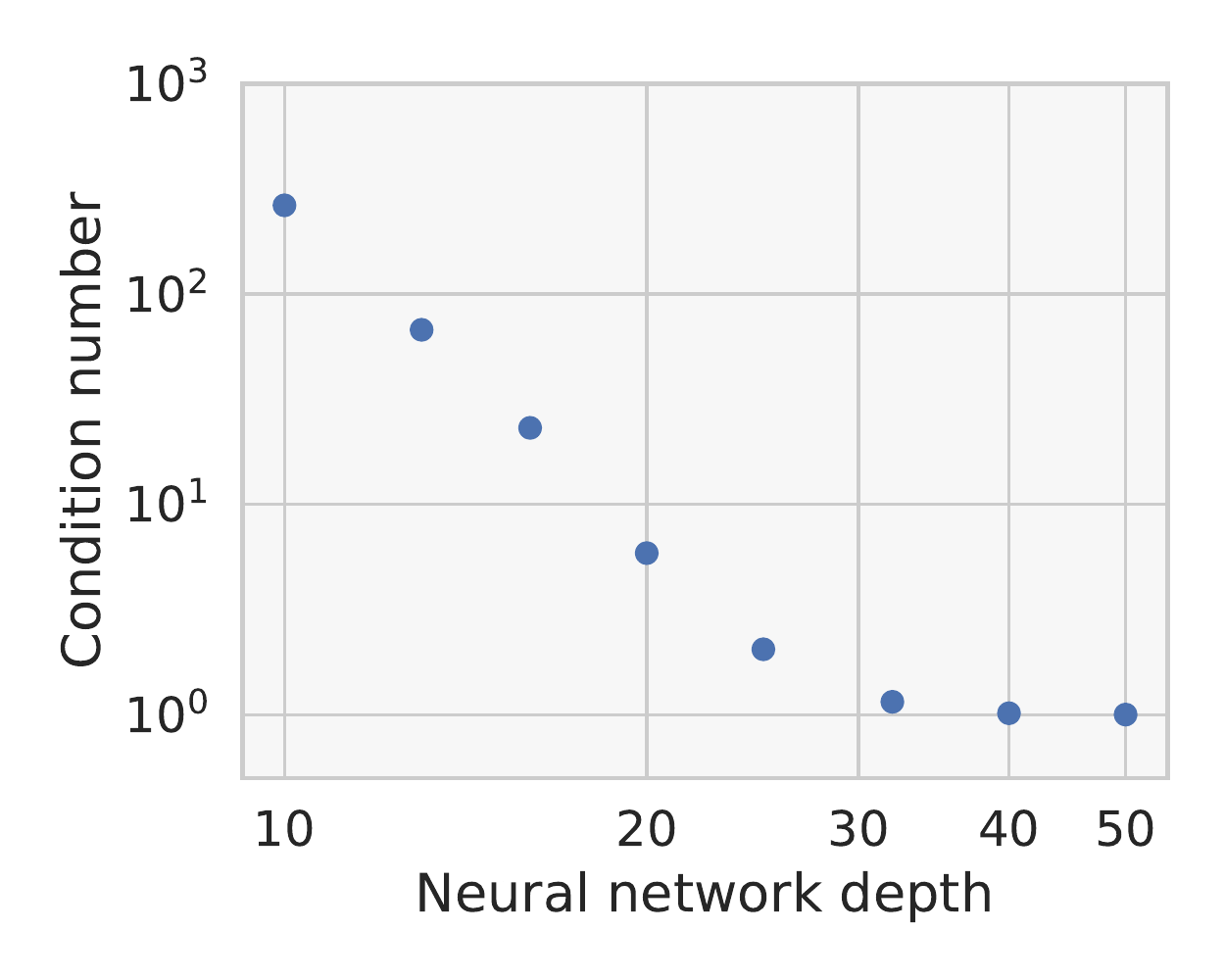}
  \caption{Fully-connected NTK ($n=128$)\label{fig:my:cond:fc}}
  \end{subfigure}%
  \begin{subfigure}[t]{0.33\textwidth}
  \includegraphics[width=\textwidth]{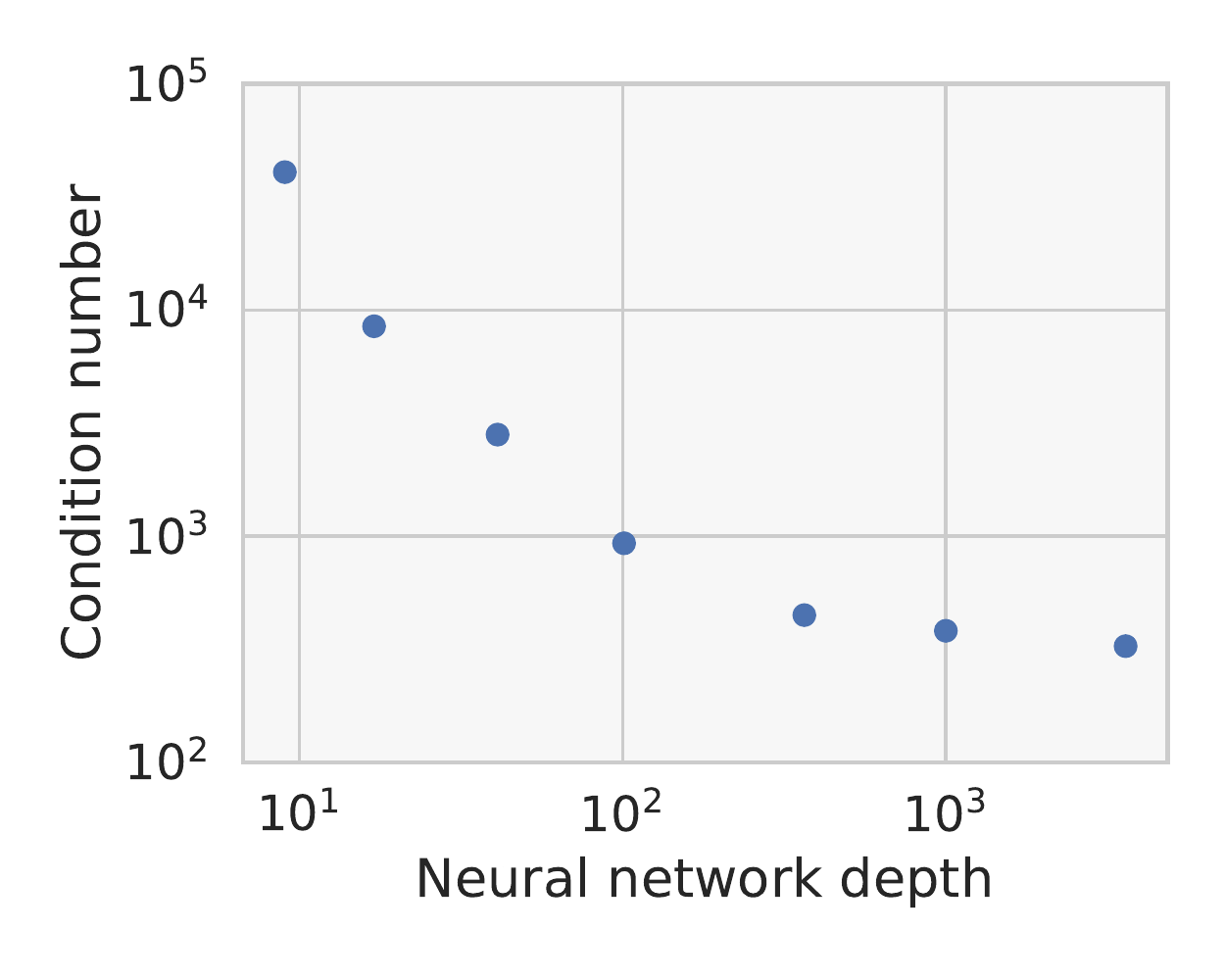}
  \caption{Convolutional NTK ($n=128$)\label{fig:my:cond:cnn}}
  \end{subfigure}%
  \begin{subfigure}[t]{0.33\textwidth}
  \includegraphics[width=\textwidth]{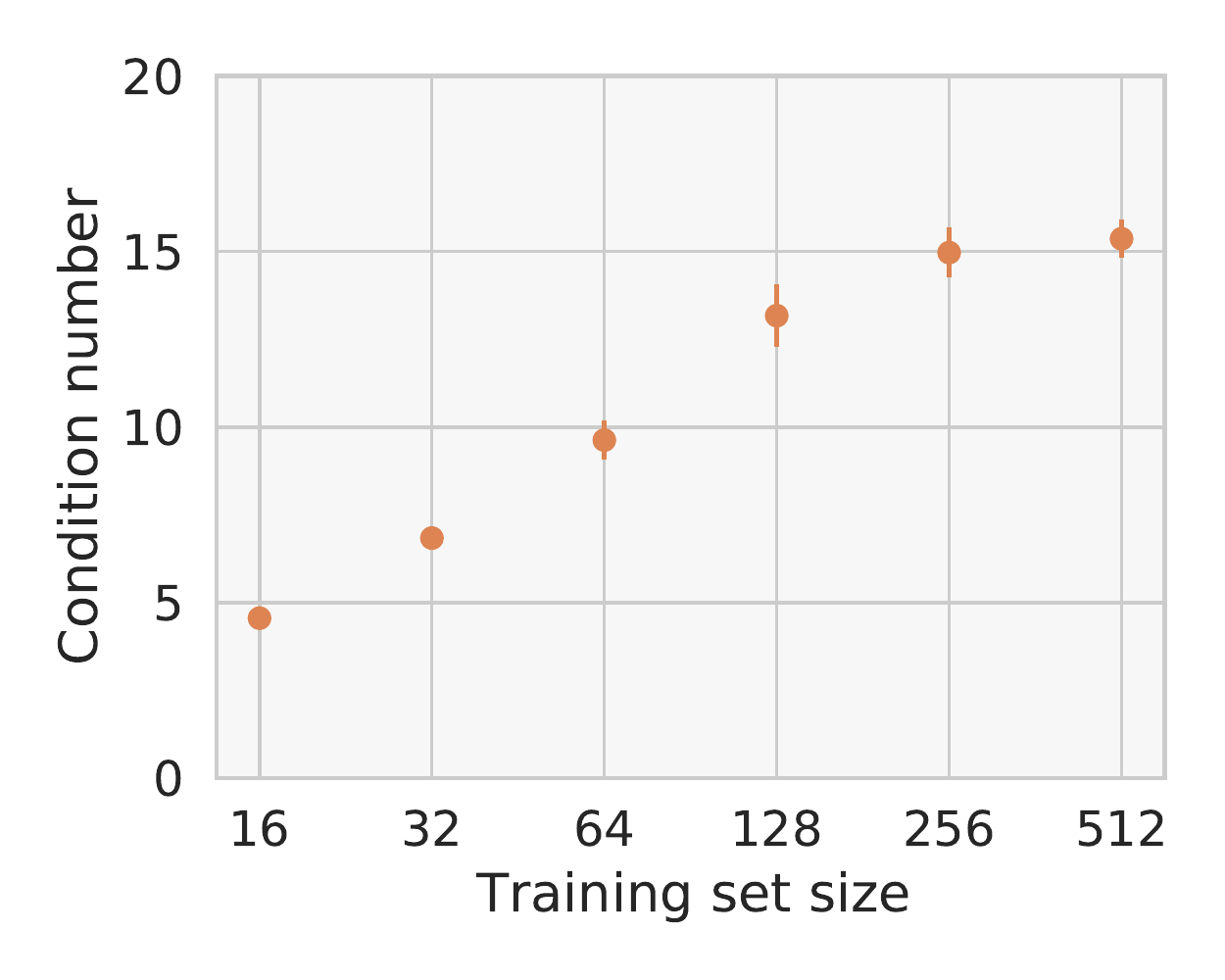}
  \caption{Sparsified convolutional NTK with fixed neural network depth ($L=101$)\label{fig:my:cond_var}}
  \end{subfigure}
  \caption{\textit{Conditioning of the NTK as a function of depth and training set size.} \textbf{(a)} and \textbf{(b)} While improved conditioning with increasing depth is only proven for the fully-connected neural network, the convolutional neural network shows similar behavior; as expected, the condition number approaches a constant not equal to one. \textbf{(c)} For fixed depth $L=101$, well-conditioning $\kappa = O(\log n)$ is efficiently achieved by exaggerating the vanishing of off-diagonal elements according to the Gershgorin circle theorem.}
\end{figure}

Due to its properties as a kernel, the NTK matrix of the convolutional neural network is positive definite. However, due to the truncation of matrix elements, the sparsified NTK may acquire negative or zero eigenvalues. For a fully-connected network at sufficient depth to be well-conditioned, the vanishing off-diagonals ensured that setting the matrix elements to zero did not significantly perturb the matrix; in the regime of shallower neural networks, however, the matrix must be preconditioned. As seen in Figures~\ref{fig:my:el} and~\ref{fig:my:cond:cnn}, the simulated depth of 101 layers used in the main results is likely not as deep as the appropriate $L_\mathrm{conv}$ for the convolutional neural network. Consistent with the natural NTK structure, preconditioning can thus be performed by exaggerating the vanishing off-diagonals with a multiplicative factor. Hence, the properties of an NTK corresponding to a neural network efficiently trainable by gradient descent are exploited to successfully precondition the sparsified NTK without significantly impacting classification performance.

By the Gershgorin circle theorem, the eigenvalues of the NTK are within a given radius of the diagonal elements; this radius is the sum of the magnitudes of the non-diagonal elements of a matrix row. In the numerical experiments shown in Figure~\ref{fig:my:perf} of the main text, we increase the training set size $n$ with a sparsification pattern enforcing $s = O(\log n)$ while holding the network depth at $L = 101$. This corresponds to increasing the Gershgorin eigenvalue bound by $O(\log n)$. Accordingly, given numerical verification of the well-conditioning of a preconditioned NTK at small $n$, the off-diagonals of the NTK at larger $n$ can be suppressed by a factor scaling like $O(\log n)$. In practice, the Gershgorin bound is looser than the true eigenvalue bound; we find that it suffices to suppress the off-diagonals by $O((\log n)^{9/10})$. The resulting conditioning is shown in Figure~\ref{fig:my:cond_var}.

The procedure $\mathcal{P}_A$ accessed by the QLSA uses the same sparsity pattern as the fully-connected network, ensuring $s = O(\log n)$ (previously shown in Figure~\ref{fig:ff:sparse}). Since introducing more off-diagonal elements can cause the matrix to become increasingly ill-conditioned, it is helpful to bias the nonzero NTK elements towards the most important examples in the training dataset. Since the NTK takes larger values for more similar data, nonzero elements selected by the sparsity pattern are additionally set to zero if their magnitude is below a given percentile off-diagonal element, as estimated by evaluating a subset of the training set ($n = 16$ for the reported numerical experiments). Since training examples are selected i.i.d., such a statistic is efficient to estimate with few samples. Note that this procedure does not affect the time complexity of training the neural network, as the same fixed sparsity pattern is used regardless of the choice of neural network.

Since the matrix sparsity is chosen to scale like $O(\log n)$ and the condition number is empirically observed to be bounded by $O(\log n)$, the QLSP corresponding to the convolutional NTK is efficient to solve with a quantum linear systems algorithm. As discussed in the main text, the measurements required for post-selection and readout are also bounded by $O(\log n)$ for the MNIST data set, yielding a total runtime that is polylogarithmic in training set size to train the convolutional neural network. This demonstration of a convolutional neural network resembling common deep learning models provides numerical evidence for the generality of an exponential speedup over gradient descent.

\end{document}